\documentclass[journal,twoside,web]{ieeecolor}

\usepackage{generic}
\usepackage[english]{babel}
\usepackage{cite}
\usepackage{amsmath,amssymb,amsfonts}
\usepackage{algorithm}
\usepackage{algorithmic}
\usepackage{graphicx}
\usepackage{textcomp}
\usepackage{amsmath}
\usepackage{dsfont}
\usepackage{color}
\usepackage{units}
\usepackage{booktabs}
\usepackage{comment}
\usepackage{url}
\usepackage[normalem]{ulem}
\usepackage[hidelinks]{hyperref}
\usepackage{caption}
\usepackage{subcaption}

\usepackage{amsthm}
\newcounter{thm}
\newtheorem{prob}[thm]{Problem}
\newtheorem{lemma}{Lemma}[section]

\newtheorem{theorem}{Theorem}[section]
\newtheorem{definition}{Definition}[section]
\newtheorem{remark}{Remark}[section]
\newtheorem{approximation}{Approximation}[section]

\newcommand{\flexbrac}[1]{\if\relax\detokenize{#1}\relax \else (#1) \fi}
\newcommand{\flexcomma}[1]{\if\relax\detokenize{#1}\relax \else ,#1 \fi}

\newcommand{\abs}[1]{\lvert #1 \rvert}
\newcommand{\cA}{\mathcal{A}}

\newcommand{\cC}{\mathcal{C}}

\newcommand{\cG}{\mathcal{G}}

\newcommand{\cM}{\mathcal{M}}

\newcommand{\cR}{\mathcal{R}}

\newcommand{\cV}{\mathcal{V}}

\DeclareMathOperator*{\argmin}{arg\,min}

\def\BibTeX{{\rm B\kern-.05em{\sc i\kern-.025em b}\kern-.08em
   T\kern-.1667em\lower.7ex\hbox{E}\kern-.125emX}}

\newif\ifmargincomments %A quick way of turning off margin comments for, say, arXiv submission
\margincommentstrue

\ifmargincomments

\else

\fi

\begin{document}
\title{A Time-invariant Network Flow Model for Ride-pooling in Mobility-on-Demand Systems}

\author{Fabio Paparella, Leonardo Pedroso, Theo Hofman, Mauro Salazar
\thanks{Date of submission \today. This publication is part of the
project NEON with number 17628 of the research
program Crossover, partly financed by the Dutch
Research Council.}
\thanks{ Fabio Paparella, Leonardo Pedroso, Theo Hofman and Mauro Salazar are with the Control Systems Technology section, Department of Mechanical Engineering, Eindhoven University of Technology (TU/e), Eindhoven, 5600 MB, The Netherlands, (e-mail:  \tt\footnotesize \{f$\!$.$\!$paparella,l$\!$.$\!$pedroso,t$\!$.$\!$hofman,m$\!$.$\!$r$\!$.$\!$u$\!$.$\!$salazar\}@tue.nl)$\!$.}}

\maketitle
\thispagestyle{plain}
\pagestyle{plain}

\begin{abstract}
This paper presents a framework to incorporate ride-pooling from a  mesoscopic point of view, within time-invariant network flow models of Mobility-on-Demand systems. 
The resulting problem structure remains identical to a standard network flow model, 
a linear problem, which can be solved in polynomial time for a given ride-pooling request assignment.
In order to compute such a ride-pooling assignment, we devise a polynomial-time knapsack-like algorithm that is optimal w.r.t.\ the minimum user travel time instance of the original problem.
Finally, we conduct two case studies of Sioux Falls and Manhattan, where we validate our models against state-of-the-art time-varying results, and we quantitatively highlight the effects that maximum waiting time and maximum delay thresholds have on the vehicle hours traveled, overall pooled rides and actual delay experienced.
%We show that the higher the intensity of demands per unit time, the higher the performance of the systems, i.e., lower waiting time and delay experienced by users. 
We show that for a sufficient number of requests, with a maximum waiting time and delay of 5 minutes, it is possible to ride-pool more than 80\% of the requests for both case studies. Last, allowing for four people ride-pooling can significantly boost the performance of the system.% For a sufficiently large number of travel demands, a significant portion of them can be pooled in batches of three or four people.
\end{abstract}


% !TeX spellcheck = en_US

\section{Introduction}\label{sec:intro}

\begin{figure}[t]
	\centering
	\begin{subfigure}{\linewidth}
		\centering
		\includegraphics[width=.8\linewidth]{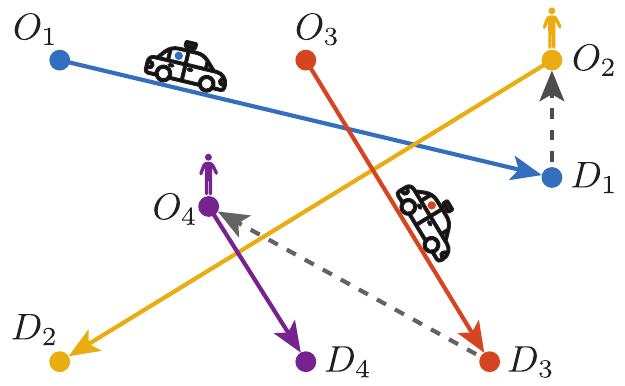}
		\caption{Ride-sharing.}
	\end{subfigure}\\
	\vspace{0.5cm}
	\begin{subfigure}{\linewidth}
		\centering
		\includegraphics[width=.8\linewidth]{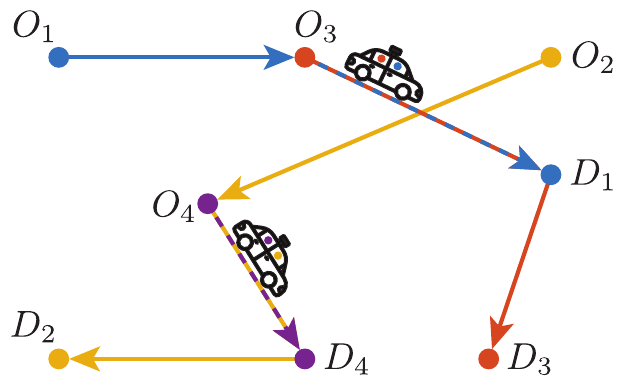}
		\caption{Ride-pooling.}
	\end{subfigure}
	\caption{Illustrative comparison of ride-sharing and ride-pooling with two vehicles; gray dashed lines represent rebalancing flows without users aboard. }
	\label{fig:eg_ridepooling}
\end{figure}

Mobility-on-Demand is transforming urban transportation.
While, ride-sharing is the concept of having only one user at the time on each vehicle, ride-pooling is the feature of allowing two or more users traveling contemporaneously on each vehicle at lower costs, e.g., emissions, energy consumption, fleet size, and price. Nevertheless, these improvements are paid in terms of additional waiting time and delays caused by detours. An illustrative comparison between ride-pooling and ride-sharing is depicted in Fig.~\ref{fig:eg_ridepooling}. Hence one of the main challenges is to understand how to balance these positive and negative aspects to achieve a solution that works for all the stakeholders. 
Ride-pooling is a difficult problem to deal with due to its inherently combinatorial nature that scales with the number of travel requests. Microscopic problems like ride-pooling are incompatible with problems on a different scale. However, sometimes it is enough to study a ride-sharing system from a mesoscopic point of view---whose complexity is usually independent from the number of travel requests---, such as when studying mobility for planning or design~\cite{SalazarLanzettiEtAl2019,LukeSalazarEtAl2021}. In this paper, we devise a framework to study ride-pooling from a mesoscopic perspective leveraging a stochastic approach. Such framework allows to incorporate ride-pooling into a time-invariant multi-commodity network flow model, also known as traffic flow model, a mesoscopic modeling framework commonly used for a variety of mobility planning and design scopes~\cite{ZardiniLanzettiEtAl2022} and whose complexity is independent from the number of travel requests. Then, leveraging flow decomposition algorithms \cite{AhujaMagnantiEtAl1993}, microscopic vehicle routes can be reconstructed.

% For this reason, the microscopic nature of ride-pooling  is, at first sight, incompatible with an approach on a different scale. In this paper, we propose a framework to deal with ride-pooling from a mesoscopic point of view by moving from a deterministic to a stochastic approach. In particular, we devise a framework to easily incorporate ride-pooling into a linear time-invariant multi-commodity network flow model, also known as traffic flow model, that is a mesoscopic modeling framework usually used for mobility planning and design. 

\textit{Related Literature:} 
%This paper pertains to the research streams of traffic flow models and ride-pooling, reviewed below.
Multi-commodity network flow models are commonly used to characterize and control ride-sharing systems, and are suited for practical implementation of a variety of constraints of different nature and can be efficiently solved with commercial solvers. This model has been used for multiple design purposes such as minimizing fleet size \cite{PavoneSmithEtAl2012,Rossi2018}, minimizing electricity costs~\cite{RossiIglesiasEtAl2018b}, smart charging~\cite{TuranTuckerEtAl2019}, and joint optimization with public transport~\cite{SalazarLanzettiEtAl2019} and the power grid~\cite{SpieserTreleavenEtAl2014,IglesiasRossiEtAl2018,ZhangPavone2018,EstandiaSchifferEtAl2021}. For example,  in \cite{LukeSalazarEtAl2021,PaparellaChauhanEtAl2023} the authors proposed a joint optimization framework for the siting and sizing of the charging infrastructure for an electric ride-sharing system. Yet all these models deal with ride-sharing systems without the possibility of ride-pooling. 

The literature in ride-pooling is very rich. 
Alonso-Mora et al.~\cite{Alonso_Mora_2017} developed the vehicle group assignment algorithm, which can solve in an optimal way the ride-pooling problem with high capacity vehicles in a microscopic setting.
In~\cite{Santi2014,JintaoHaiEtAl2020} the authors analyzed the benefits of pooling, and the pricing and equilibrium in on-demand ride-pooling markets, respectively.
Fieldbaum et al.~\cite{FielbaumBaiEtAl2021} studied ride-pooling taking into account that users can walk before and after being picked-up/dropped-off, while in~\cite{FielbaumKucharskiEtAl2022} they examined how to split costs between users that pool together. In~\cite{TsaoMilojevicEtAl2019} a time-variant network flow model was devised to compute the optimal routing for ride-pooling. However, in all of these papers the ride-pooling problem is studied from a microscopic perspective. Each request is considered individually, and as a consequence, the complexity of the resulting problems scales w.r.t.\ the number of travel requests considered.%\rev{Recently, in \cite{AouadSaritac2022} an interesting step towards a mesoscopic stance has been carried out from a stochastic matching perspective.}
%	Therein, the problem is formulated as an infinite-horizon continuous-time Markov decision process that either minimizes the expected average costs incurred by the matching decisions or maximizes the expected average rewards.}

In conclusion, to the best of the authors' knowledge, a mesoscopic time invariant  ride-pooling network flow model has not yet been proposed.

\textit{Statement of Contributions:} The main contributions of this paper are threefold. First, we propose a framework to capture ride-pooling, a microscopic combinatorial phenomenon, in a time-invariant network flow model. Crucially, this is a mesoscopic approach whereby the complexity of the problem is independent of the number of travel requests and that can be implemented in already existing time-invariant network flow models.
% a mesoscopic approach where %the arrival process is stochastic and 
%the complexity of the problem is independent from the number of requests. 
Second, within the proposed framework, we devise a method to compute a ride-pooling request assignment that is optimal w.r.t.\ the minimum user travel time problem, a particular case of the minimum vehicle travel time problem.
Third, we showcase our framework with two case studies of Sioux Falls, USA, and Manhattan, NYC, USA, where we show that the inclusion of ride-pooling can significantly benefit ride-sharing mobility systems.

A preliminary version of this paper will be presented at the 2023 IEEE Conference on Decision and Control~\cite{PaparellaPedrosoEtAl2023}.
In this extended version, we carry out a broader literature review, extend the  approach from two to an arbitrary number of $K$ people ride-pooling, carry out a complexity analysis of the problem, and compare the quality of the solution obtained w.r.t.\ the granularity of the road network.
Moreover, we conduct a real-world case study of Sioux Falls, USA, and one for Manhattan, USA. In addition, the results of Manhattan are compared with the offline oracle model (all the requests are known in advance) of Santi et al.~\cite{Santi2014}, a microscopic model whose complexity depends on the number of travel requests.

\textit{Organization:} The remainder of this paper is structured as follows: Section~\ref{sec:model} introduces the multi-commodity traffic flow problem and the framework to capture ride-pooling. Section~\ref{sec:res} details the case studies of Sioux Falls and Manhattan.
Last, in Section~\ref{sec:conc}, we draw the conclusions from our findings and provide an outlook on future research endeavors.

\emph{Notation:}  We denote the vector with all elements equal to 1, of appropriate dimensions, by $\mathds{1}$. The $i$th component of a vector $v$ is denoted by $v_i$ and the entry $(i,j)$ of a matrix $A$ is denoted by $A_{ij}$. The cardinality of set $\mathcal{S}$ is denoted by $\abs{\mathcal{S}}$. Given a bag $\mathcal{C}$, its support is denoted by $\mathrm{Supp}(\mathcal{C})$, the multiplicity of element $i$ is denoted by $m_\mathcal{C}(i)$, and its cardinality is denoted by $|\mathcal{C}|:= \sum_{i\in\mathrm{Supp}(\mathcal{C})} m_\mathcal{C}(i)$.
% !TeX spellcheck = en_US

\section{Ride-pooling Network Flow Model}\label{sec:model}

In this section, we introduce the network  traffic flow model~\cite{PavoneSmithEtAl2012}. Then, we extend it to take into account ride-pooling, and finally present a brief discussion of the model.

%From the vehicle flows can then be reconstructed  reconstruct 

% \rev{Leveraging flow decomposition algorithms \cite{AhujaMagnantiEtAl1993}, the microscopic vehicle routes can be reconstructed}

\subsection{Time-invariant Network Flow Model} 

We model the mobility system as a multi-commodity network flow model, similar to the approaches of~\cite{SalazarLanzettiEtAl2019,LukeSalazarEtAl2021,PaparellaChauhanEtAl2023,RossiIglesiasEtAl2018b,PaparellaSripanhaEtAl2022}. The transportation network is a directed graph $\mathcal{G} = (\mathcal{V}, \mathcal{A})$. {It consists of a set of vertices $\mathcal{V} := \{1,2,...,\abs{\mathcal{V}}\}$,} representing the location of intersections on the road network, and {a set of arcs $\mathcal{A} \subseteq \mathcal{V} \times \mathcal{V}$}, representing the road links between {intersections}. We indicate ${B \in \{-1,0,1\}^{\abs{\cV} \times \abs{\cA}}}$ as the incidence matrix~\cite{Bullo2018} of the road network $\mathcal{G}$. Consider an arbitrary arc indexing of natural numbers $\{1,\ldots,\abs{\cA}\}$, then $B_{ia} = -1$ if the arc indexed by $a$ is directed towards vertex $i$, $B_{ia} = 1$ if the arc indexed by $a$ leaves vertex $i$,  and $B_{ia}=0$ otherwise. We denote $t$ as the vector whose entries are the travel time $t_{a}$ required to traverse each arc $a\in \cA$, ordered in accordance with the arc ordering of $B$. % which we assume to be \msmargin{constant}{? wrt?}. 
Similarly to \cite{PavoneSmithEtAl2012}, we define travel requests as follows:
\begin{definition}[Requests]
	A travel request is defined as the tuple $r = (o,d,\alpha) \in \mathcal{V} \times \mathcal{V} \times \mathbb{R}_{>0}$, in which $\alpha$ is the number of users traveling from the origin $o$ to the destination $d \neq o$ per unit time.  Define the set of requests as $\mathcal{R} := \{r_m\}_{m\in \mathcal{M}}$, where $\mathcal{M} = \{1,\ldots,M\}$.
\end{definition} 

We assume, without any loss of generality, that the origin-destination pairs of the requests $r_m \in \mathcal{R}$ are distinct. 
In this paper, we distinguish between active vehicle flows, which correspond to the flows of vehicles serving users whether they are ride-pooling or not, and rebalancing flows which correspond to the flows of empty vehicles between the drop-off and pick-up vertices of consecutive requests. We define the active vehicle flow induced by all the requests that share the same origin $i \in {\mathcal{V}}$ as vector $x^{i}$, where element $x^{i}_{a}$ is the flow on arc $a \in \cA$, ordered in accordance with the arc ordering of $B$. The overall active vehicle flow is a matrix $X \in \mathbb{R}^{\abs{\cA} \times \abs{\cV}}$ defined as $X := \left[x^{1}\  x^2 \, \dots \,x^{\abs{\cV}}\right]$, where we assume every node is a potential origin. The rebalancing flow across the arcs is denoted by $x^{\mathrm{r}} \in \mathbb{R}^{\abs{\cA}}$. In the following, we define the network flow problem.
\begin{prob}[Multi-commodity Network Flow Problem]\label{prob:main}
	Given a road graph $\cG$ and a set of travel requests $\mathcal{R}$, the active vehicle flows $X$ and rebalancing flow $x^\mathrm{r}$ that minimize the cost in terms of overall travel time result from
	\begin{equation*}
		\begin{aligned}
			\min_{X, x^\mathrm{r}}\; &J(X,x^\mathrm{r}) =t^\top ( X \mathds{1} + \rho x^\mathrm{r} )  \\
			\mathrm{s.t. }\; & BX = D \\
			&B ( X \mathds{1}+ x^\mathrm{r} )=0 \\
			& X, x^\mathrm{r} \geq 0,
		\end{aligned}
	\end{equation*}
	where $\rho$ is a weighting factor, and the demand matrix $D \in \mathbb{R}^{\abs{\mathcal{V}} \times \abs{\mathcal{V}}}$ represents the requests between every pair of vertices, whose entries are
	\begin{equation}\label{eq:def_D}
		\!\!\!\!D_{ij} = \begin{cases}
			\alpha_m, &  \exists m \in \mathcal{M} : o_m = j \land d_m = i\\
			-\sum_{k\neq j} D_{kj}, & i  = j \\ %\sum \limits_{\substack{k = 1 \\ k\neq i}}^{\abs{\mathcal{V}}} D_{ik}
			0, &   \mathrm{otherwise}.%i\neq j \land \nexists m \in \mathcal{M} : o_m = i \land d_m = j\\
		\end{cases}\!\!\!
	\end{equation}
\end{prob}
For $\rho = 0$, Problem~\ref{prob:main} is totally unimodular, and $X$ and $x^\mathrm{r}$ can be decoupled and computed separately~\cite{Rossi2018}. In this case the objective function can be interpreted as the minimum user travel time. For $\rho =1$ the objective is the vehicle minimum travel time, that is equivalent to the minimum fleet size problem, i.e., the minimum number of vehicles required to implement the flows~\cite{PavoneSmithEtAl2012,Rossi2018}. 

%  

% !TeX spellcheck = en_US

\subsection{Ride-pooling Time-invariant Network Flow Model}

In this paper, we propose a formulation to take into account ride-pooling of, at most, $K$ requests without the need to change the original structure of the problem. We transform the original set of requests, portrayed by $D$, into an equivalent set of requests accounting for ride-pooling, portrayed by $D^\mathrm{rp}$. We define the ride-pooling network flow problem as follows:
\begin{prob}[Ride-pooling Network Flow Problem]\label{prob:rides}
	Given a road graph $\cG$ and a demand matrix $D^\mathrm{rp}$, the active vehicle flows $X$ and rebalancing flow $x^\mathrm{r}$ that minimize the cost in terms of overall travel time result from
	\begin{equation*}
		\begin{aligned}
			\min_{X,x^\mathrm{r}}\; &J(X,x^\mathrm{r}) = {t^\top (X \mathds{1} + \rho x^r)} \\
			\mathrm{s.t. }\; & BX = D^\mathrm{rp} \\
			&B ( X \mathds{1}+ x^\mathrm{r} )= 0 \\
			& X, x^\mathrm{r} \geq 0.
		\end{aligned}
	\end{equation*}
\end{prob}

The ride-pooling demand matrix $ D^\mathrm{rp}$ in Problem~\ref{prob:rides}, which describes the pooling pattern, has to be determined according to four key conditions. First, the individual requests, described by $D$, must be served. Second, ride-pooling $k\leq K$ requests is only spatially feasible if the detour travel time of every user is not greater than a threshold $\bar{\delta} \in \mathbb{R}_{ \geq 0}$. Third, ride-pooling $k\leq K$ requests is only temporally feasible if the maximum waiting time for a request to start being served does not exceed a threshold $\bar{t} \in \mathbb{R}_{>0}$. Fourth, the requests are pooled to minimize the cost function of Problem~\ref{prob:rides} at its solution. Due to the combinatorial nature of such an endeavor, we relax the problem in order to devise a computationally tractable algorithm, according to the following approximation:

%such that: i)~the individual requests, described by $D$, are served, and ii)~the requests are pooled to minimize the cost function of Problem~\ref{prob:rides} at its solution. 

\begin{approximation}\label{approx}
For the purpose of computing the demand matrix $D^\mathrm{rp}$, we set $\rho = 0$ and the cost function of Problem~\ref{prob:rides} becomes $\tilde{J}(X):=t^\top X \mathds{1}$.
\end{approximation}

It is crucial to remark that this approximation is only employed in a first step to compute $D^\mathrm{rp}$ and, afterwards, in a second stage, the passenger and rebalancing flows are optimized jointly according to Problem~\ref{prob:rides} with weight $\rho =1$. This is in order for an Autonomous Mobility-on-Demand fleet, for example, which is centrally operated and whose vehicles do not compete for rides.  %In fact, in the simulations performed in Section~\ref{sec:res}, $t^\top x^\mathrm{r}$ accounts for less than $3\%$ of $J$ for every scenario studied. 
It is a very common approximation used in the literature \cite{Alonso_Mora_2017}, which we leverage to devise a polynomial-time algorithm to compute $D^\mathrm{rp}$ that is optimal w.r.t.\ the approximated version of the problem. 

%This approximation makes sense in the context of the problem, since the active vehicle travel time $t^\top X \mathds{1}$ is usually dominant over the rebalancing travel time $t^\top x^\mathrm{r}$. In fact, in the simulations performed in Section~\ref{sec:res}, $t^\top x^\mathrm{r}$ accounts for less than $3\%$ of $J$ for every scenario studied. Crucially, leveraging Approximation~\ref{approx}, we can devise a polynomial-time algorithm to compute $D^\mathrm{rp}$ that is optimal w.r.t. the approximated version of the problem.

%allows    which we show in the sequel that can be computed in ploynomia--ime .  , as shown in the sequel.
%Thereby, to compute such an assignment, we devise a polynomial-time algorithm that is optimal w.r.t.\ an approximated version of the problem.
%Note that it is not exact, but it makes sense
%the cost function of Problem~\ref{prob:rides} to $\tilde{J}(X) := {t^\top X \mathds{1}}$. 
%say that it makes sense (teh relaxation)
%unlike Problem~\ref{prob:main},
%%Compared to Problem~\ref{prob:main}, in Problem~\ref{prob:rides}, each vehicle can transport more than one user at the same time. Thus, the vehicles' travel time is lower or equal to the users'. By using $J(X,x^\mathrm{r}) =t^\top (X  \mathds{1} +  x^\mathrm{r})$, we obtain the minimum fleet size.
%Polynomial time, we show in the sequel. 
%Makes sense

\subsection{Approximate  Computation of the Demand Matrix}

In this section, we present a framework to compute the demand matrix $D^\mathrm{rp}$ under Approximation~\ref{approx}.  

% we ormalize the four conditions and present a alg which is optimal in relation to j tilde

\subsubsection{Spatial Analysis of Ride-pooling}\label{sec:SD}
%TO BE ADDED BEFORE THE TIME|SPACE sections\msmargin{In a ride pooling environment, both space and time are relevant to understand if \textcolor{blue}{two} requests can be pooled. However, in Problem~\ref{prob:main}, the elements of the OD matrix $b$ are coefficients of a Poisson process (i.e., \textcolor{blue}{it models} requests per unit time). For this reason, we must move from a deterministic point of view to a probabilistic one. In particular, we decouple time properties of the requests from the spacial ones: We first compute the optimal ride pooling paths, and then, we evaluate the probability of them happening. We define the following ride pooling problem.}{dobbiamo rendere questo paragrafo pi\`u chiaro: \`e cruciale per il paper}
We analyze the feasibility and optimal configuration of ride-pooling $k\leq K$ requests from a spatial perspective. First, we define $\delta$ as the delay experienced by each user, representing the time required to travel the additional detour distance w.r.t.\ the scenario without ride-pooling. If the delay experienced by any of the $k$ users is higher than the threshold $\bar{\delta}$, then pooling the $k$ requests is unfeasible. 

Second, given the feasible pooling itineraries, we analyze which one is the optimal, i.e., the best itinerary to serve the requests, and whether pooling is advantageous compared to no pooling. Let $C_k(\mathcal{M})$ denote all combinations with repetition of $k$ elements of $\mathcal{M}$, thus an element $\mathcal{C} \in C_k(\mathcal{M})$ is a bag of request indices, i.e., a relaxed concept of a set that allows for repeated elements. For instance, we denote the bag that contains $a,b\in \mathbb{N}$ once and $c\in \mathbb{N}$ twice by $[a,b,c,c]$. Therefore, $\bigcup_{k = 1}^K C_k(\mathcal{M})$ is the set of all bags of at most $K$ requests that can be ride-pooled. Consider a bag of requests to ride-pool $\mathcal{C} \in \bigcup_{k = 1}^K C_k(\mathcal{M})$. To restrict this analysis to the spatial dimension, we temporarily make two key considerations, that we lift in Section~\ref{sec:STF}: i)~ all requests  $r_n$ with $n\in\mathcal{C}$ are made at the same time; and ii)~all requests have the same demand, which we set, without any loss of generality, to $\alpha=\unit[1]{requests/unit \; time}$. %$\alpha \in \mathbb{R}^+$.
\begin{figure}[t]
	\centering
	\includegraphics[width = 0.65\linewidth]{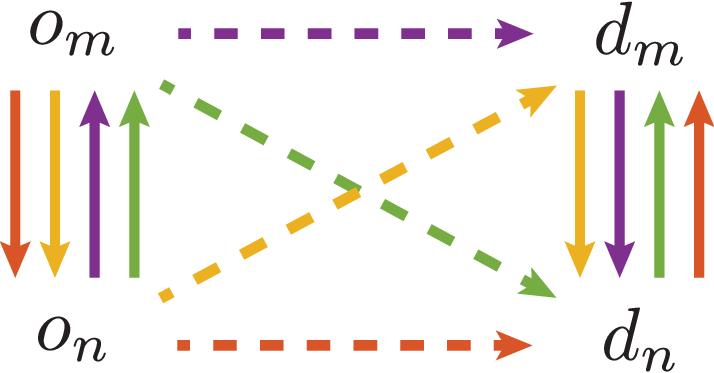}
	\caption{Distinct configurations for serving a bag $\mathcal{C} = [n,m]$ of two requests $r_m,r_n \in \cR$. Each arrow represents a flow of $\alpha = 1$ vehicles. The dashed arrows represent a flow with two users, whilst the solid ones represent a flow with one user~\cite{PaparellaPedrosoEtAl2023}.}
	\label{fig:serve_conf}
	\vspace{-0.5cm}
\end{figure}
There are several different ways to serve a bag of requests in $\mathcal{C}$. Each corresponds to a sequence of origins and destinations of the requests in $\mathcal{C}$ such that i)~for each request the corresponding origin is visited before the destination; and ii)~no edge carries a null user flow. Denote the set of such sequences by $\mathcal{S}_{\mathcal{C}}$.

To illustrate the modeling framework, consider a bag of two requests $r_m,r_n \in \mathcal{R}$, which is denoted as $\mathcal{C} = [m,n] = [n,m] \in C_2(\mathcal{M})$. In this specific case, there are four different sequences of serving these requests, as depicted in Fig.~\ref{fig:serve_conf}. Note that, for instance, the sequence $(o_m,d_m,o_n,d_n) \notin \mathcal{S}_{\mathcal{C}}$ is not valid because the edge from $d_m$ to $o_n$ has a null user flow. In fact, two bags of a single request, namely $[n]$ and $[m]$, cover the same user flows without the need for an edge with null user flow. On the contrary, for example, $(o_m,o_n,d_m,d_n)\in \mathcal{S}_{\mathcal{C}}$ because there is always at least one user in the vehicle. 

%There are five different ways of serving two requests $r_m, r_n \in \cR$, as depicted in Fig.~\ref{fig:serve_conf}.
The goal is to assess whether it is feasible to ride-pool the requests in a bag $\mathcal{C} \in \bigcup_{k = 1}^K C_k(\mathcal{M})$ and find the best sequence of nodes to serve them. Note that each sequence that serves a bag $\mathcal{C} \in C_k(\mathcal{M})$ can be split into at most $2k-1$ travel requests. For example, the sequence $(o_m,o_n,d_n,d_m)$ represented in orange in Fig.~\ref{fig:serve_conf} can be split into three vehicle flow requests: $(o_m,o_n,\alpha = 1)$,  $(o_n,d_n,\alpha = 1)$, and  $(d_n,d_m,\alpha = 1)$. Note that during $(o_n,d_n,\alpha = 1)$ there are two users per vehicle request because they ride-pool, whereas in the remainder there is only one. For a bag $\mathcal{C}$ and a sequence $s\in \mathcal{S}_{\mathcal{C}}$, denote the set of such equivalent vehicle flow requests by $\mathcal{R}^s_\mathcal{C}$ and the corresponding demand matrix, given by \eqref{eq:def_D}, by $D^{\mathcal{C},s}$. Solving Problem~\ref{prob:rides}, under Approximation~\ref{approx},  with a simplified demand matrix  $D^\mathrm{rp} = D^{\mathcal{C},s}$ we obtain a vehicle flow $X^{\mathcal{C},s} \in \mathbb{R}^{\abs{\mathcal{V}}\times \abs{\mathcal{V}}}$, which is equivalent to a shortest path procedure. Additionally, to characterize the baseline without ride-pooling, denote by $t_m^0$ the time it takes to serve request $r_m$ with $m\in \mathcal{C}$ individually, which can also be easily computed with a shortest path procedure. Now, ride-pooling a bag $\mathcal{C}$ with a sequence $s\in \mathcal{S}_{\mathcal{C}}$ incurs a delay of serving request $r_m$ with $m\in \mathcal{C}$ of
\begin{equation*}
	\delta^{\mathcal{C},s}_m = \sum_{p \in {\pi^{\mathcal{C},s}_{m}} } [t^\top X^{\mathcal{C},s}]_p   - t_m^0 %[t^\top X^{\mathcal{C},0}]_{o_m},
\end{equation*}%
w.r.t. no ride-pooling, where $\pi^{\mathcal{C},s}_{m}$ is the sequence of nodes contained in $s$ starting at $o_m$ and ending in the node before $d_m$. Thus, ride-pooling a bag $\mathcal{C}$ is feasible from a spatial perspective if there is $s\in\mathcal{S}_\mathcal{C}$ such that $\delta^{\mathcal{C},s}_m \leq \bar{\delta}, \,\forall m \in \mathcal{C}$. Denote by $C_k^{\bar{\delta}}$ the set of feasible bags of $k$ requests. If feasible, the optimal sequence is the one that achieves the lowest cost among the feasible ones, i.e., 
\begin{equation*}
	\begin{split}
	\mathcal{S}^\star_{\mathcal{C}} := \argmin_{s\in \mathcal{S}_\mathcal{C}} \quad & \tilde{J}(X^{\mathcal{C},s})\\
	\mathrm{s.t.}  \quad & \delta^{\mathcal{C},s}_m \leq \bar{\delta}, \,\forall m \in \mathcal{C} .%\max_{m\in \mathcal{C}} 	\delta^{\mathcal{C},s}_m \leq \bar{\delta}.
	\end{split}
\end{equation*}
Finally, we define the optimal demand matrix of a feasible bag of requests $\mathcal{C}$ as $D^{\mathcal{C},\star} = D^{\mathcal{C},s}$ with some $s\in \mathcal{S}^\star_\mathcal{C}$ chosen randomly.
%\begin{equation*}
%	D^{\mathcal{C},\star} := 
%	\begin{cases}
%		D^{mn,s}, \quad & s\in \mathcal{S}^\star_\mathcal{C}\\
%		D^{mn,0}, \quad & \mathcal{S}^\star_\mathcal{C} = \emptyset\\
%	\end{cases}
%\end{equation*}

%Additionally, to characterize the baseline without ride-pooling, define $\mathcal{R}^0_\mathcal{C}:= \{r_n \in \mathcal{R}  : n\in \mathcal{C}\}$, which corresponds to individually serving the unique requests in $\mathcal{C}$. \textcolor{red}{[Here I consider unique requests]} Likewise,  $D^{\mathcal{C},0}$ and $X^{\mathcal{C},0}$ are defined from $\mathcal{R}^0_\mathcal{C}$. Now, ride-pooling a bag $\mathcal{C}$ with a sequence $s\in \mathcal{S}_{\mathcal{C}}$ the delay of serving request $r_m$ with $m\in \mathcal{C}$ inccurs a delay of

\begin{remark} 
	A detailed computational complexity analysis of the computation of $D^{\mathcal{C},\star}$ for all $\mathcal{C} \in \bigcup_{k = 1}^K C^{\bar{\delta}}_k(\mathcal{M})$ is carried out in Appendix~\ref{sec:comp_complex}. It is shown that a simple algorithm has a computational complexity that grows with $\mathcal{O}((2K!)K |\mathcal{V}|^{2K+2})$ with $K$ and $|\mathcal{V}|$. Notice that the meaningful values of $K$ are generally very small and that, for a fixed and bounded $K$, the complexity is polynomial w.r.t. the dimension of the network. Furthermore, this computation is very easy to parallelize and has to be carried out only once.
\end{remark}

\subsubsection{Temporal Analysis of Ride-pooling}\label{sec:TD}

In this section, we analyze the temporal alignment of $k$ requests for ride-pooling.
We derive the probability of $k$ requests taking place within the maximum waiting time, $\bar{t}$. As common in traffic flow models~\cite{PavoneSmithEtAl2012}, we consider that the arrival rate of a request $r_m \in \cR$ follows a Poisson process with parameter $\alpha_m$. In the following lemma, we indicate the probability of $k$ events occurring within a maximum time window $\bar{t}$.

\begin{lemma}\label{lemma:finalprob}
	Let $r_{1}, \ldots, r_{k} \in \mathcal{R}$ be $k$ requests whose arrival rates follow independent Poisson processes with parameters $\alpha_{1}, \ldots, \alpha_{k}$, respectively. The probability of all having an occurrence within a maximum time interval $\bar{t}$ is
	\begin{equation}\label{eq:lem}
		P_{\bar{t}}\left(\alpha_{1}, \ldots, \alpha_{k}\right)=\sum_{i=1}^{k} \frac{\alpha_{i}}{\sum_{j=1}^{k} \alpha_{j}} \prod_{\substack{j=1\\j\neq i}}^{k}\left(1-e^{-\alpha_{j} \bar{t}}\right).
	\end{equation}
\end{lemma}
\begin{proof}
	The proof can be found in Appendix~\ref{app:1}.
\end{proof}

Consider $k$ requests in a bag $\mathcal{C}\in C_k(\mathcal{M})$. 
Henceforth, the probability of all requests in $\mathcal{C}$ having an event within a maximum time window $\bar{t}$ is given by $P_{\bar{t}}(\alpha_i, i\in \mathcal{C})$.

\subsubsection{Expected Number of Pooled Rides}
\label{sec:STF}
In Section~\ref{sec:SD}, we analyzed the spatial dimension of the ride-pooling problem, whereby we computed the best feasible pooling path given a bag of up to $K$ requests. In Section~\ref{sec:TD}, we analyzed the temporal dimension of the ride-polling problem, whereby we derived the probability of all requests in a bag of up to $K$ requests taking place within a time window $\bar{t}$.  By lifting the temporary assumptions made in Section~\ref{sec:SD}, namely i) the requests are made at the same time and ii) $\alpha=\unit[1]{requests/unit \; time}$, we formulate the ride-pooling demand matrix given a certain pooling assignment, defined in the following:

Given a bag of requests $\mathcal{C} \in \bigcup_{k = 1}^K C^{\bar{\delta}}_k(\mathcal{M})$, a fraction of the demand of those requests can be assigned to be pooled together according to a sequence $s^\star\in \mathcal{S}^\star_\mathcal{C}$, as discussed in Section~\ref{sec:SD}. Consider that each request $r_m$ with $m\in \mathcal{C}$ has some demand flow $\alpha^\prime_m$ that is assigned to be ride-pooled in $\mathcal{C}$. Notice that the maximum pooled vehicle demand for requests $\mathcal{C}$ is $\min(\alpha^\prime_i/m_{\mathcal{C}}(i), i\in \mathcal{C})$, which would be exactly achieved for an infinite waiting time. It is interesting to remark that one must account for the multiplicity of the elements in $\mathcal{C}$, since if the same request $r_i$ with $i\in \mathcal{C}$ is pooled with itself with a multiplicity of $m_{\mathcal{C}}(i)$, then the maximum vehicle flow for that request is naturally the fraction of the original person demand by the number of people served by the same vehicle. However, from the analysis in Section~\ref{sec:TD}, becomes evident that only a fraction of that demand can actually be pooled due to the aforementioned temporal constraints. Specifically, the probability of pooling bag $\mathcal{C}$ is given by $P_{\bar{t}}(\alpha_i^\prime, i\in \mathcal{C})$ according to Lemma~\ref{lemma:finalprob}. Therefore, the effective expected pooled demand of $\mathcal{C}$  is given by $ \gamma_\mathcal{C} := {\min(\alpha^\prime_i/m_{\mathcal{C}}(i), i\in \mathcal{C})}P_{\bar{t}}(\alpha_i^\prime, i\in \mathcal{C})$.  As a result, according to the spatial analysis in Section~\ref{sec:SD}, this pooled demand is portrayed by the demand matrix $\gamma_{\mathcal{C}}D^{\mathcal{C},\star}$. The full ride-pooling demand matrix $D^\mathrm{rp}$ is made up of two contributions: i)~the sum of the expected pooled active vehicle flows of the form $\gamma_{\mathcal{C}}D^{\mathcal{C},\star}$ for $\mathcal{C} \in \bigcup_{k = 1}^K C^{\bar{\delta}}_k(\mathcal{M})$; and ii)~the requested demands that were not ride-pooled.  Thus, the entry $(i,j)$ of $D^\mathrm{rp}$ can be written as
\begin{equation*}\label{eq:b}
		\!D_{ij}^{\mathrm{rp}} \! = \!\begin{cases}
			\alpha_m \! + \!\!\!\!  \!\!\! \sum\limits_{{\mathcal{C} \in \bigcup \limits_{k = 1}^K \!\! C^{\bar{\delta}}_k(\mathcal{M}) \; :\; m\notin \mathcal{C}}}\!\!\!\! \!\!\!\!  \!\!\!\!  \!\!\!\!   \gamma_{\mathcal{C}}D_{ij}^{\mathcal{C},\star}, &\exists m \!\in\! \mathcal{M}\! :\! (d_m,o_m)\! = \!(i,j) \\[2.5em]
			- \sum_{k\neq j} D_{kj}^{\mathrm{rp}},  &i  = j \\[1em]
			\sum\limits_{{\mathcal{C} \in \bigcup \limits_{k = 1}^K \!\! C^{\bar{\delta}}_k(\mathcal{M})}} \!\!\!\! \!\!\!\!  \gamma_{\mathcal{C}}D_{ij}^{\mathcal{C},\star},  &\text{otherwise}.
		\end{cases}
\end{equation*}

Finally, given the assigned demand flows $\alpha^\prime_m$ for $r_m$ that can be ride-pooled in $\mathcal{C}$, one can input $D^\mathrm{rp}$ to Problem~\ref{prob:rides}, which yields an LP. In the next section, the optimal choice for the assignment demands $\alpha^\prime_m$ is derived.

\subsubsection{Optimal Ride-pooling Assignment}\label{sec:alg}

In this section, we will compute the optimal ride-pooling assigned demands $\alpha^\prime_m$ for $r_m$ that can be ride-pooled in $\mathcal{C}$ and the corresponding $\gamma_\mathcal{C}$, under Approximation~\ref{approx}, leveraging an iterative approach, which is described in what follows. For every bag of requests $\mathcal{C} \in \bigcup_{k = 1}^K  C^{\bar{\delta}}_k(\mathcal{M})$, we can compute the unitary improvement of the objective function of Problem~\ref{prob:rides}, denoted by $\Delta \tilde{J}_{\mathcal{C}}$, w.r.t.\ the no-pooling scenario. To characterize the baseline without ride-pooling, define  the demand matrix $D^{\mathcal{C},0}$ which is a simplified demand matrix that characterizes the individual (not necessarily unique) requests in $\mathcal{C}$ with unitary flow. Solving Problem~\ref{prob:rides}, under Approximation~\ref{approx},  with a simplified demand matrix  $D^\mathrm{rp} = D^{\mathcal{C},0}$ we obtain a vehicle flow $X^{\mathcal{C},0} \in \mathbb{R}^{\abs{\mathcal{V}}\times \abs{\mathcal{V}}}$, which again amounts to a graph search procedure. Then,  $\Delta \tilde{J}_{\mathcal{C}}$ amounts to the difference  between $\tilde{J}(X^{\mathcal{C},\star})$ and $\tilde{J}(X^{\mathcal{C},0})$.
	
We now propose an algorithm to compute the optimal ride-pooling assignment. Let $\alpha_m^\prime, m\in \cM$ represent the demand of request $r_m$ that has not yet been assigned, which is updated every iteration of the algorithm and is initialized as $\alpha_m^\prime = \alpha_m$. Each iteration, the bag of requests with the highest relative improvement  w.r.t.\ the user flow is prioritized with the highest possible pooling demand assignment. That is, in each iteration, if $\mathcal{C} \in \bigcup_{k = 1}^K C^{\bar{\delta}}_k(\mathcal{M})$  is the pair of requests with the highest $\Delta \tilde{J}_{\mathcal{C}}/|\mathcal{C}|$, we assign $\alpha_m^\prime, m\in \mathcal{C}$ as the demand to be pooled in $\mathcal{C}$. Moreover, the rides that have been assigned but not pooled, are added back to the original requests, i.e., we set $\alpha_m^\prime  \leftarrow  \alpha_m^\prime-\gamma_{\mathcal{C}}m_\mathcal{C}(m), \;\forall m\in \mathrm{Supp}(\mathcal{C})$.
%Let $\Delta \tilde{J}_{\mathcal{C}}^\prime$ denote another auxiliary variable  throughout the iterations, initialized as $\Delta \tilde{J}_{\mathcal{C}}^\prime=  \Delta \tilde{J}_{\mathcal{C}}$. At the end of every iteration, $\Delta \tilde{J}^\prime_{\mathcal{C}}$ is set to zero.
This procedure is repeated until convergence is achieved, i.e., $\max_{\mathcal{C} \in \bigcup_{k = 1}^K  C^{\bar{\delta}}_k(\mathcal{M})} (\Delta \tilde{J}^\prime_{\mathcal{C}}/|\mathcal{C}|) \leq 0$. The pseudocode of this procedure is presented in Algorithm~\ref{alg:one}. In the following theorem, we establish the convergence and optimality of Algorithm~\ref{alg:one}.

\begin{algorithm}[t] 
	\caption{Compute optimal assignment flows $\gamma_\mathcal{C}$.}\label{alg:one}
	\begin{algorithmic}
		
		\STATE $\Delta \tilde{J}_{\mathcal{C}} \; \leftarrow \; \tilde{J}(X^{\mathcal{C},\star})- \tilde{J}(X^{\mathcal{C},0}),  \quad \forall \mathcal{C} \in \bigcup_{k = 1}^K  C^{\bar{\delta}}_k(\mathcal{M})$ 
%		\STATE $\tilde{J}_{mn} \leftarrow  \mathrm{input} \;\; D^{mn, \star} \; \text{to Problem}~\ref{prob:rides}, \;\; \forall {m,n\in \cM}$
%		\STATE $\tilde{J}_{m} + \tilde{J}_n  \leftarrow  \mathrm{input} \;\; D^{mn, 0} \; \text{to Problem}~\ref{prob:rides}, \;\; \forall {{m,n\in \cM}}$
%%		\STATE $ \forall m,n \; \mathrm{input} \; D^{mn,0} \; to \; Problem~\ref{prob:rides} \rightarrow \tilde{J}_{m},\tilde{J}_{n}$
%		\STATE $\Delta \tilde{J}_{mn} \;\leftarrow\; \tilde{J}_{m} + \tilde{J}_{n} - \tilde{J}_{mn} $
%		\STATE $\Delta \tilde{J}_{mn}^\prime \;\leftarrow\; \Delta \tilde{J}_{mn}, \; \forall m,n\in \cM$
		\STATE $\alpha_{m}^\prime \;\leftarrow\; \alpha_{m}, \quad \forall m\in \cM $
%		\STATE $\;D^\mathrm{rp} = \mathbb{0}$
		\WHILE {$\max_{\mathcal{C} \in \bigcup_{k = 1}^K  C^{\bar{\delta}}_k(\mathcal{M})} (\Delta \tilde{J}_{\mathcal{C}}/|\mathcal{C}|) > 0$}
		\STATE $ \mathcal{C} \in \mathrm{argmax}_{\mathcal{C} \in \bigcup_{k = 1}^K  C^{\bar{\delta}}_k(\mathcal{M})} (\Delta \tilde{J}_{\mathcal{C}}/|\mathcal{C}|)  \phantom{l^{l^{l}}}$
		\STATE $ \gamma_\mathcal{C} = \min\left(\alpha^\prime_m/m_{\mathcal{C}}(m), m\in \mathcal{C}\right) P_{\bar{t}}(\alpha_m^\prime, m\in \mathcal{C})\phantom{l^{l^{l^{l}}}}$
%		\IF{$o_n = o_m \; \mathrm{and} \; d_n=d_m$}
%		\STATE $ \beta_{mn}  \;\leftarrow\; \alpha_{m}^\prime, \;\beta_{nm} \;\leftarrow\; \beta_{mn}$
%		\STATE $ \gamma_{mn} \;\leftarrow\; \beta_{mn}  P(\beta_{mn},\beta_{nm})/2, \;\gamma_{nm} \;\leftarrow\; \gamma_{mn}$
%%		\STATE $D^\mathrm{rp}_{ii} = b^\mathrm{rp}_{ii} - \beta_{ijkl} , \; b^\mathrm{rp}_{ij} = b^\mathrm{rp}_{ij} + \beta_{ijkl}  $
%		\ELSE
%		\STATE $\beta_{nm}  \;\leftarrow\;  \alpha_{n}^\prime, \; \beta_{mn}  \;\leftarrow\;  \alpha_{m}^\prime$
%		\STATE $\gamma_{mn}  \;\leftarrow\;  \min(\beta_{nm},\beta_{mn}) P(\beta_{mn},\beta_{nm})$
%		\STATE $\gamma_{nm} \;\leftarrow\;  \gamma_{mn}$
%		\ENDIF
%		\STATE $D^\mathrm{rp} = D^\mathrm{rp} + \gamma_{mn}  D^{mn \star}$
%		\STATE$\alpha_{m}^\prime  \;\leftarrow\;  \alpha_{m}^\prime - \gamma_{mn},  \; \alpha_{n}^\prime  \;\leftarrow\; \alpha_{n}^\prime - \gamma_{nm} $
%		\STATE$\Delta \tilde{J}^\prime_{mn}  \;\leftarrow\; 0, \;   \Delta \tilde{J}^\prime_{nm} \leftarrow \Delta \tilde{J}^\prime _{mn}  $
		\STATE $\alpha_{m}^\prime  \;\leftarrow\;  \alpha_{m}^\prime - m_\mathcal{C}(m)\gamma_{\mathcal{C}} \quad \forall m\in \mathrm{Supp}(\mathcal{C})  \phantom{l^{l^{l^{l^l}}}} $
		\STATE $\Delta \tilde{J}_{\mathcal{C}} \; \leftarrow \; 0   \phantom{l^{l^{l^{l^l}}}}$
		\ENDWHILE
	\end{algorithmic}
\end{algorithm}
%We define the optimal objective function of Problem~\ref{prob:rides} with parametric entry $\gamma$ as $\tilde{J}(X^\star_\gamma)$.

\begin{theorem}\label{theorem:one}
Let $X^\star_\gamma$ denote the optimal solution of Problem~\ref{prob:rides}, under Approximation~\ref{approx}, for the effective ride-pooling allocation $\gamma_\mathcal{C}, \; \mathcal{C} \in \bigcup_{k = 1}^K  C^{\bar{\delta}}_k(\mathcal{M})$. Then, in $\abs{\cM}(\abs{\cM}^K-1)/(\abs{\cM}-1)$ iterations at most, Algorithm~\ref{alg:one} converges to a minimizer of $\tilde{J}(X^\star_\gamma)$ among all valid effective ride-pooling allocation patterns.
\end{theorem} 
\begin{proof}
		The proof can be found in Appendix~\ref{app:1}. 
\end{proof}

\subsection{Discussion}

A few comments are in order. First, the mobility system is analyzed at steady-state in a time-invariant framework, which is unsuitable for an online implementations, but it has been used for planning and design purposes by several works in the literature as seen in Section~\ref{sec:intro}. This assumption is reasonable if the travel requests vary slowly w.r.t.\ the average time of serving each request. This is the case especially in highly populated metropolitan areas \cite{IglesiasRossiEtAl2017,Neuburger1971}.
Second, our framework does not take into account the stochastic nature of the exogenous congestion that determines that travel time in each road arc. However, this deterministic approach is suitable for our purposes as it provides an average representation of these stochastic phenomena on a mesoscopic scale~\cite{Neuburger1971}.
Third, Problems~\ref{prob:main} and \ref{prob:rides} allow for fractional flows, which is acceptable because of the mesoscopic perspective of the work~\cite{RossiZhangEtAl2017,SalazarLanzettiEtAl2019}. Randomized sampling methods can be leveraged to compute integer flows from fractional ones, achieving near-optimality~\cite[Ch.~4]{Rossi2018}. Moreover, new requests can be implemented through a receding horizon framework. Finally, $D^\mathrm{rp}$ is optimal only w.r.t.\ the objective function of Problem~\ref{prob:rides} with $\rho =0$, enabling a polynomial-time computation.
%\input{Sections/Granularity.tex}
% !TeX spellcheck = en_US
\section{Case Studies}\label{sec:res}
This section showcases our modeling and optimization framework in two case studies of Sioux Falls, USA and Manhattan, USA. In the first case study, conducted in Sioux Falls, the simulation is carried out for $K = 4$, i.e., four people ride pooling. In the second case study, we run a simulation for $K=2$, leveraging the original road network of $\abs{\cV}= 357$ nodes. Then, we compare the results obtained from this macroscopic model with the aggregate ones of the microscopic model obtained by~\cite{Santi2014}. 
Problem~\ref{prob:rides} is parsed with YALMIP~\cite{Loefberg2004} and solved with Gurobi 9.5~\cite{GurobiOptimization2021}.
We compute its solution considering the optimal ride-pooling assignment of the relaxed problem, obtained as described in Section~\ref{sec:alg}, for a varying amount of hourly demands, obtained by uniformly scaling the demand of the historical requests, and for various waiting times and maximum delays. 
%A MATLAB implementation of the methods presented is available in an open-source repository at {\small \url{https://github.com/fabiopaparella/ride-pooling-MoD}}.

\subsection{Sioux Falls}
In this case study, we showcase our framework for Sioux Falls, USA. 
The road network, which is depicted in Fig.~\ref{fig:SF}, and the travel requests are retrieved from the GitHub repository of Transportation Networks for Research~\cite{ResearchCoreTeam}.
\begin{figure}[t]
	\centering
	\includegraphics[trim={0cm 300 0 80},clip,width=1\linewidth]{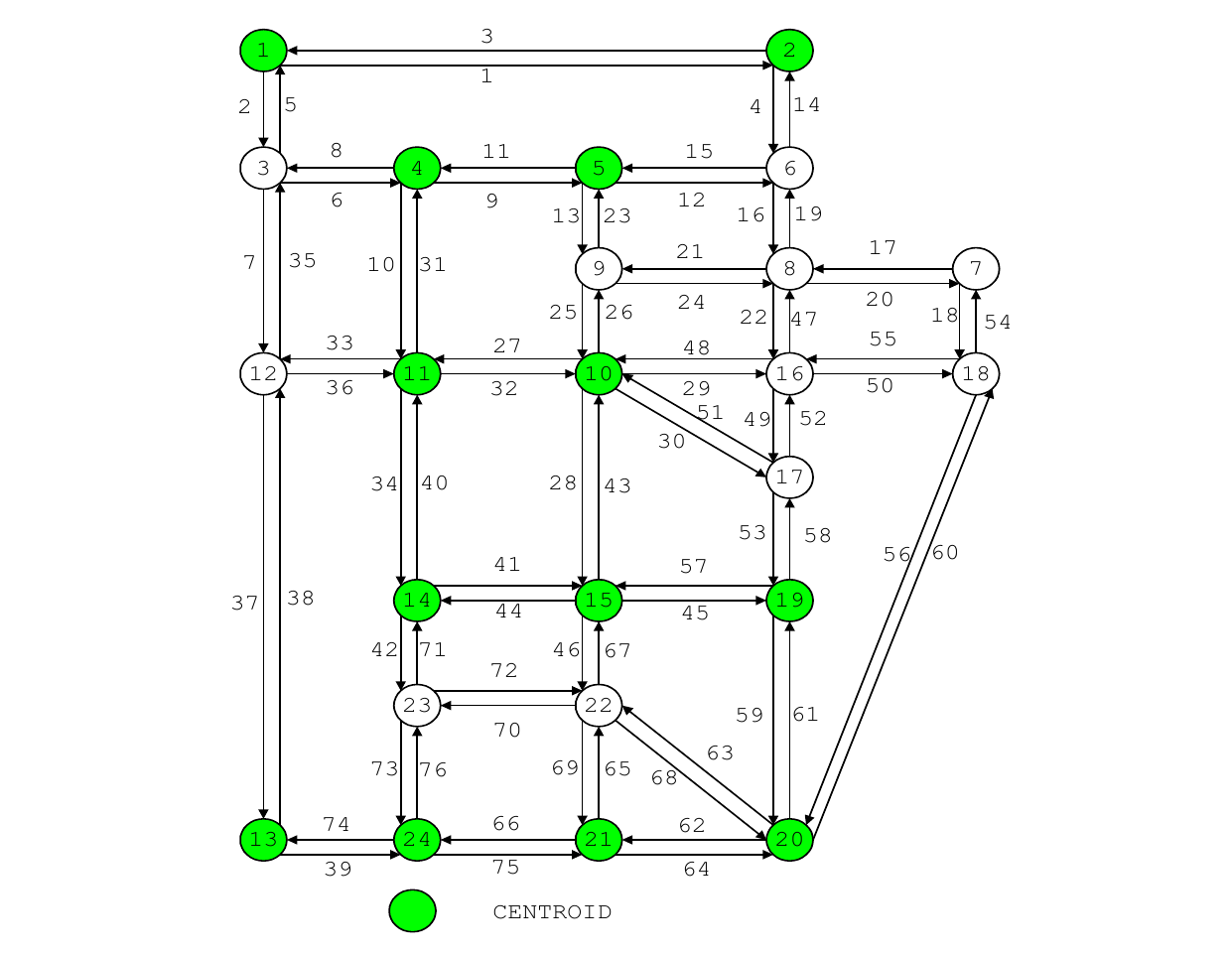}
	\caption{Road network of Sioux Falls, South Dakota, USA. Prepared by Hai Yang and Meng Qiang, Hong Kong University of Science and Technology.}
	\label{fig:SF}
	\vspace{-0.2cm}
\end{figure}
The computation of the set of matrices $D^{\cC,s}$ took \unit[5]{h} on a standard laptop computer.
Then, for each simulation, running Algorithm~\ref{alg:one} and solving Problem~\ref{prob:rides} took overall \unit[1]{min} approximately. Fig.~\ref{fig:4k} shows the number of travel requests that were ride-pooled, and the average delay experienced $\delta$---for which $\bar{\delta}$ is the upper bound---as a function of the number of overall travel requests, the maximum waiting time $\bar{t}$, and maximum delay $\bar{\delta}$.
\begin{figure}[t]
	\centering
	\includegraphics[trim={0cm 0 15 0},clip,width=\linewidth]{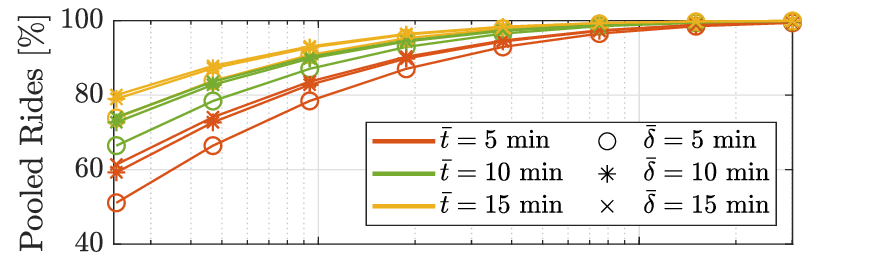}
	\includegraphics[trim={0cm 0 15 0},clip,width=\linewidth]{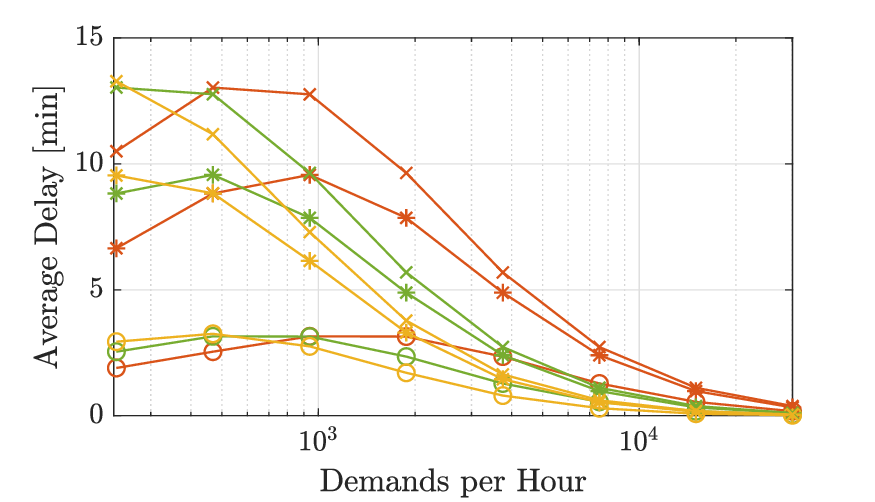}
	\caption{Ride-pooling with $k \leq 4$ in Sioux Falls. Percentage of pooled rides and average experienced delay as a function of the overall number of hourly demands, waiting time $\bar{t}$, and maximum delay $\bar{\delta}$. Note that the x-axis is logarithmic. }
	\label{fig:4k}
	\vspace{-0.1cm}
\end{figure}
Interestingly, the average delay is not always monotone decrescent. It is the case for a short waiting time $\bar{t}$, that is counter-balanced by a longer delay. On the left of the bottom figure, the delay is low because not many requests are pooled together. As the demand increases, the percentage of pooled rides increases as well, and as a consequence, the delay follows. However, for larger demands, the delay decreases, as predicted by the Mohring and Better Matching effects~\cite{FielbaumTirachiniEtAl2021}: the higher the number of travel requests in the system, the lower the delay experienced. Conversely, the fewer people use a mobility service, the poorer the performance of the system, reflecting in a lower percentage of requests that can be effectively ride-pooled. 
Fig.~\ref{fig:comp} shows the composition of rides in terms of number of people that can be ride-pooled together as a function of the number of demands and maximum waiting time.
\begin{figure}[t]
	\centering
	\includegraphics[trim={30 0 50 0},clip,width=\linewidth]{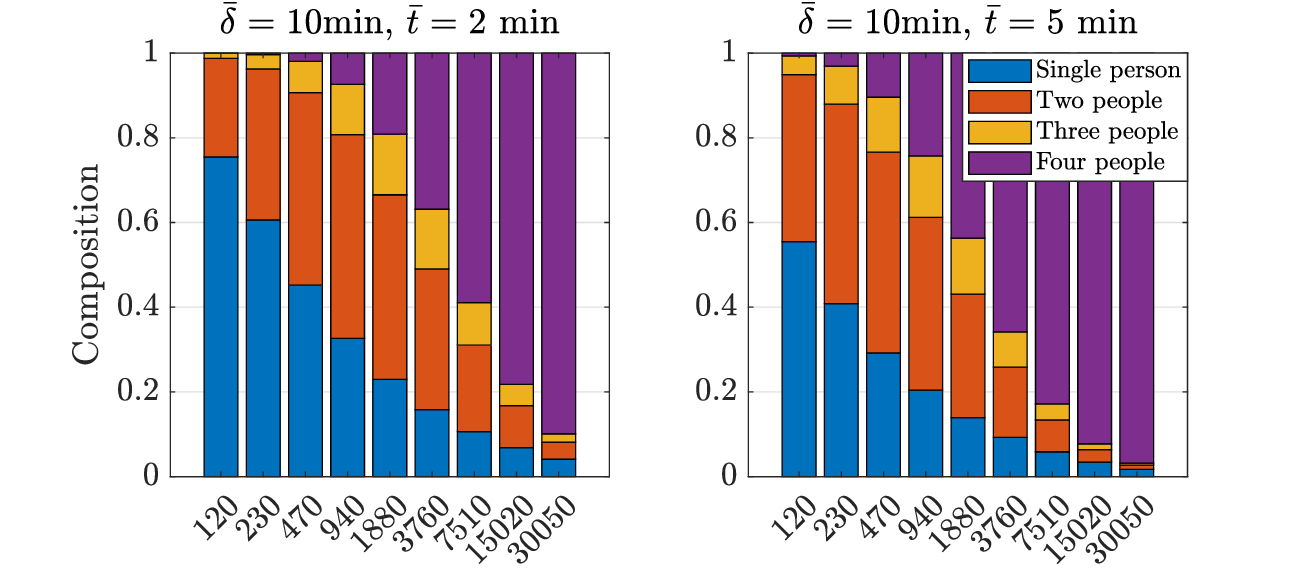}
	\includegraphics[trim={30 0 50 0},clip,width=\linewidth]{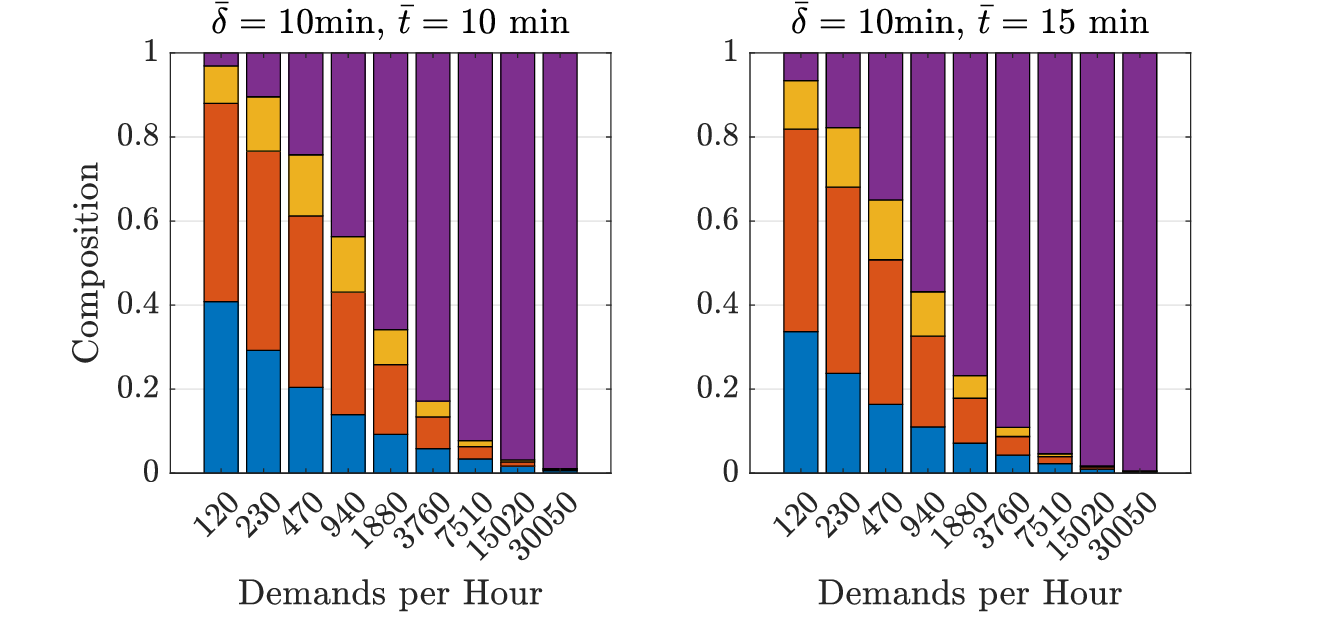}
	\caption{Composition of rides as a function of number of demands per hour, for a maximum experienced delay of $\bar{\delta }=\unit[10]{min}$ and a varying waiting time $\bar{t} = \unit[2,5,10,15]{min}$ in Sioux Falls.}
	\label{fig:comp}
	\vspace{-0.1cm}
\end{figure}
We highlight that, for low demands and short waiting times, the majority of the ride-pooling is done with two people, while three people is marginal and four people is negligible. On the contrary, towards larger number of demands, three and four people ride-pooling become more dominant. Specifically, above $15000$ hourly demands, four people ride-pooling counts for more than 85\%, even with a low waiting time allowed. 

\subsection{Manhattan}\label{sec:nyc}
In this section, we pursue a similar case study for Manhattan, NYC, USA, with data obtained from~\cite{HaklayWeber2008}. The travel requests are gathered from the New York Taxi and Limousine Commission, which includes one hour data of 53000 travel requests during peak-hour of March 2011, as in~\cite{Santi2014}.
In Fig.~\ref{fig:surf}, we compare the relative improvement in objective of Problem~\ref{prob:rides} w.r.t.\ Problem~\ref{prob:main}, i.e., the improvement of the overall travel time.
\begin{figure}[t]
	\vspace{0.5cm}
	\centering
	\includegraphics[width=\linewidth]{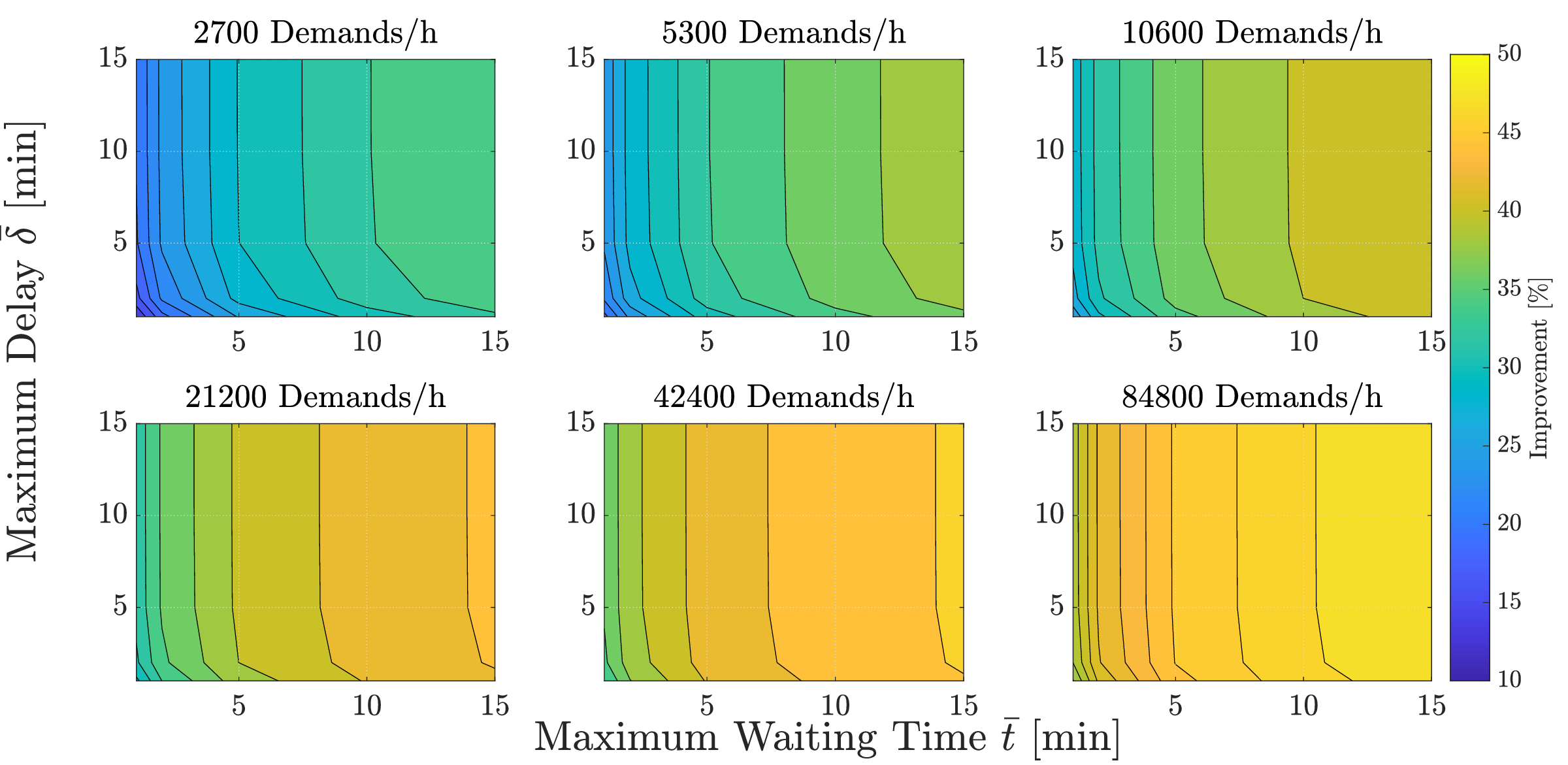}
	\caption{Improvement of the objective function of Problem~\ref{prob:rides} with $\rho=1$, i.e., vehicle travel time, w.r.t.\ ride-sharing, as a function of maximum waiting time $\bar{t}$, maximum delay $\bar{\delta}$, and demand intensity in Manhattan.}
	%		 Improvement of the objective function of Problem~\ref{prob:rides} w.r.t. the objective of Problem~\ref{prob:main}, that is the base case with no ride pooling, as a function of maximum waiting time and delay.}
\label{fig:surf}
\end{figure}
Fig.~\ref{fig:surf} shows that ride-pooling always contributes to lowering the overall travel time. In particular, the larger the number of hourly demands, the larger the relative improvement. %difference with respect to the no-pooling scenario. 
The reason is that the probability function in~\eqref{eq:lem} is monotonically increasing w.r.t.\ $\alpha_i, i\in \mathcal{C}$, i.e., the magnitude of the demands of the requests in $\mathcal{C}$.
%that, in turn, are monotonically increasing with the number of demands. 
In Fig.~\ref{fig:WD} we note that the percentage of rides that are pooled is strongly influenced by the number of demands, represented in logarithmic scale, to a lower extent by the maximum waiting time, and marginally by the maximum delay.
\begin{figure}[t]
	\vspace{0.5cm}
	\centering
	\includegraphics[trim={0cm 0 35 5},clip,width=\linewidth]{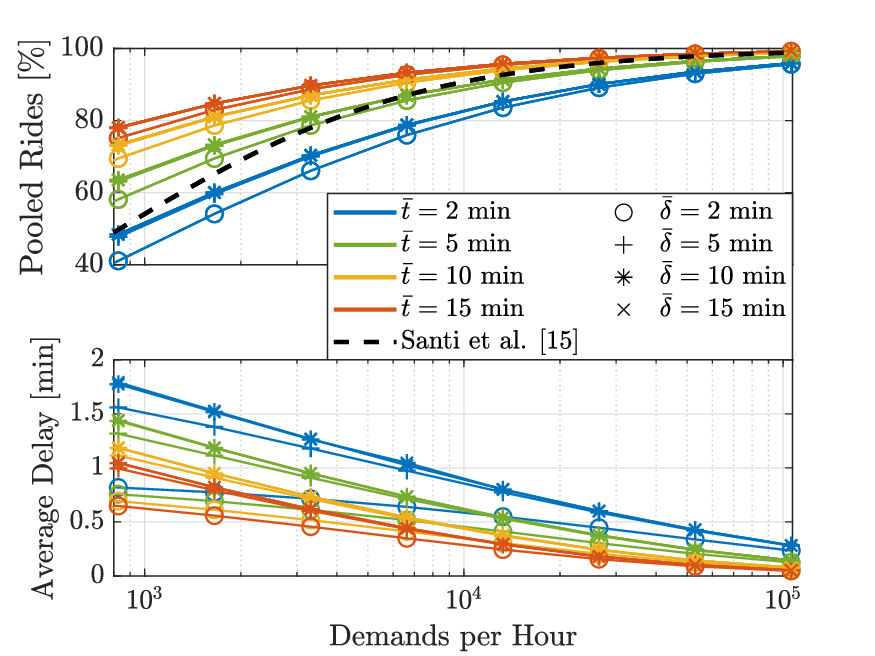}
	\caption{Percentage of pooled rides with comparison with the microscopic approach in \cite{Santi2014}, and average experienced delay as a function of the overall number of hourly demands, waiting time, and maximum delay.}
	\label{fig:WD}
	\vspace{0.5cm}
\end{figure}
The low sensitivity to the delay might be caused by the very short average distance traveled, which is \unit[2.4]{km}. 
In fact, it is always more convenient to increase waiting time rather than delay.  In addition, %In fact, 
for large demands, both the waiting time and the delay have a minor impact on the percentage of rides being pooled and on the relative improvement. %costs %, as shown in Fig.~\ref{fig:WD}.  %Then %Then, looking at  %the same figure
Moreover, for the simulations performed, $t^\top x^\mathrm{r}$, i.e., the rebalancing time, accounts for less than $7\%$ of the overall travel time for every scenario studied.  The original rebalancing time for no ride-pooling is roughly $8\%$. Thus, not only does ride-pooling decrease the overall rebalancing time due to the lower number of trips, but also does not lead to a relative increase compared to the overall travel time. This supports the hypothesis of Approximation~\ref{approx}. Last, we notice that above 20,000 hourly requests, by setting both a maximum delay and waiting time of 1 minute, it is possible to ride-pool $80\%$ of the requests. Notably, we obtain very similar qualitative and quantitative results compared to the microscopic model in~\cite{Santi2014}, which are also depicted in Fig.~\ref{fig:WD}.

\subsection{Network Granularity}\label{sec:gran}
In Section~\ref{sec:model}, we showed that the computational complexity of the method presented in this work is independent on the number of requests in the system, whilst it depends on the number of nodes in the road network. Given a road network $G=(\cV,\cA)$, the computation of $D^{\mathcal{C},\star}$ has to be executed only once, and the process is highly parallelizable.
%\msmargin{For this reason}{does not sound like logical}, for large networks the computation can be costly in terms of memory and time, is still feasible. 
However, it might be convenient to have a more tractable model, not only because of the costly computation, but also to allow for easy integration with other works that implement linear time-invariant traffic flow models~\cite{SalazarLanzettiEtAl2019,PaparellaChauhanEtAl2023}. In this section we will analyze the impact of network granularity, i.e., the number of nodes $\abs{\cV}$, on the quality of the solution. 
Given the initial road network $\cG$, we leverage a k-mean clustering method~\cite{MacQueen1967}, and obtain a pruned version, as shown in Fig.~\ref{fig:gran}. Then, we investigate the quality of the solution of the new network w.r.t.\ the original one. 
\begin{figure}[t]
	\centering
	\includegraphics[trim={1.2cm 34.6 20 5},clip,width=1\linewidth]{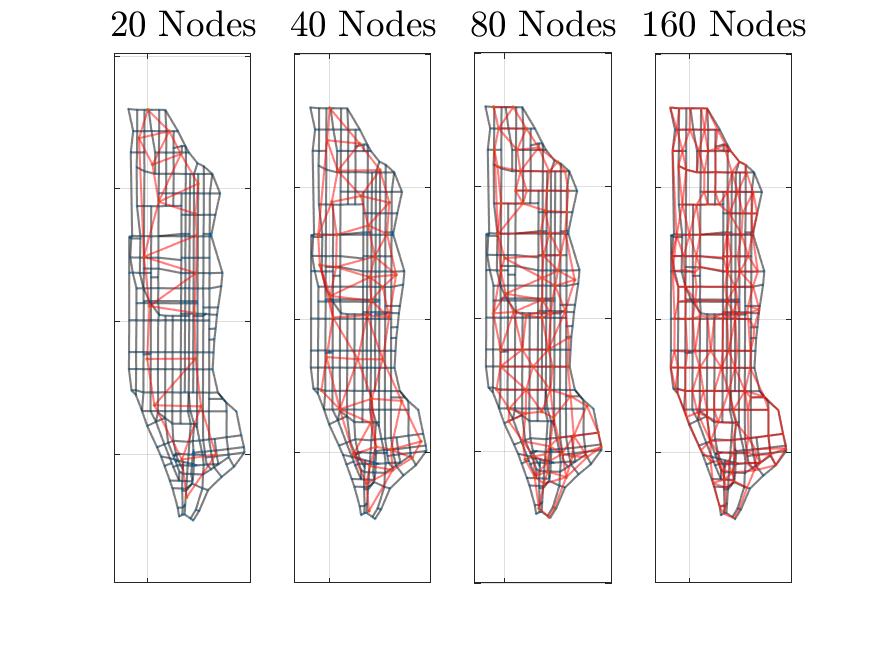}
	\caption{Pruned road networks with $\abs{\cV} = 20,40,80,160$ nodes. Networks obtained with k-mean clustering method. The length of each arc is computed as the $L_2$ norm between the two adjacent nodes. The traverse time is computed accordingly.}
	\label{fig:gran}
\end{figure}
Fig.~\ref{fig:improv} shows the relative improvement of the objective function of Problem~\ref{prob:rides} w.r.t.\ the base case (no ride-pooling) as a function on the granularity of the road network, i.e., $\abs{\cV}$.
\begin{figure}[t]
	\centering
	\includegraphics[trim={10 10 25
		10},clip,width=1\linewidth]{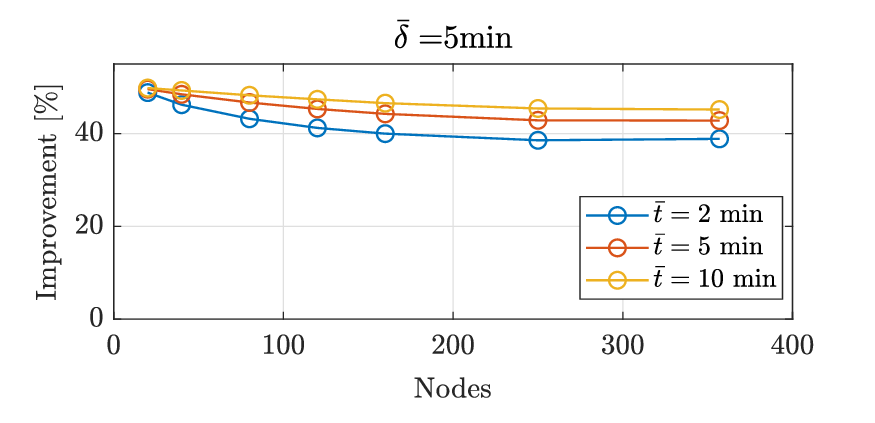}
	\caption{Improvement of the objective function of Problem~\ref{prob:rides} for the ride-pooling case w.r.t. no ride-pooling as a function of the number of nodes in the transportation network. The maximum set delay is $\bar{\delta} = \unit[5]{min}$.}
	\label{fig:improv}
\end{figure}
Naturally, the more the pruning of the road network, the more the quality of the solutions degrades, especially in cases where the new weights of the pruned edges are comparable with the maximum delay $\bar{\delta}$.  Yet, with a light pruning of the network, the solution is close to the original one. In particular, Fig.~\ref{fig:improv} shows that even by reducing $\abs{\cV}$ from $357$ to $120$-$160$, the quality of the solution is still acceptable. The computation of $D^{\cC,s}$ for $K=2$ decreases accordingly, from $\unit[90]{min}$ for the original case, down to $\unit[11]{min}$ for $\abs{\cV} = 120$.
Finally, we recall that this computation, given a network, has to be performed only once, while Algorithm~\ref{alg:one} has to be run for every case study. 
% !TeX spellcheck = en_US
\section{Conclusions}\label{sec:conc}
\color{black}
In this paper we presented a ride-pooling framework leveraging time-invariant network flow models. In particular, we do it by devising an equivalent set of requests, given a ride-pooling assignment, so that the original structure of the network flow problem remains the same.
This allows to still obtain a linear problem which can be solved in polynomial-time.
Additionally, we proposed a methodology to compute a ride-pooling request assignment, that is optimal w.r.t.\ the minimum user travel time problem. Our case studies quantitatively showed that the overall number of requests per unit time is a crucial factor to assess the benefit of ride-pooling in Mobility-on-Demand systems.
Specifically, for an increasing number of requests, an average improvement in the vehicle hours traveled of up to 45\% can be achieved.  We also showed that, for a large number of requests, more than 80\% of them could be pooled with a reasonably short waiting and delay time. Last, in dense urban environments, the majority of rides can be carried out with four people.

This work allows to include ride-pooling in research that leverages  multi-commodity network flow models. In the future, we would like to analyze congestion-aware ride-pooling, and study the interaction that ride-pooling systems have with intermodal ones.

\section*{Statement of Code Availability}
A MATLAB implementation of the methods presented is available in an open-source repository at {\small \url{https://github.com/fabiopaparella/LTI-pooling-K_people}}.

\section*{Acknowledgments}\label{Sec:akn}
We thank Dr. I. New, F. Vehlhaber, and J. van Kampen for proofreading thi paper. 
%This publication is part of the
%project NEON with number 17628 of the research
%program Crossover, partly financed by the Dutch
%Research Council.

\bibliographystyle{IEEEtran}
\bibliography{main.bib,SML_papers.bib}

% !TeX spellcheck = en_US
\appendices

% !TeX spellcheck = en_US

\section{Computation of $D^{\mathcal{C},\star}, \;\forall \mathcal{C} \in \bigcup_{k = 1}^K C^{\bar{\delta}}_k(\mathcal{M})$ }\label{sec:comp_complex}

In this section, we analyze the computational complexity of the computation of $D^{\mathcal{C},\star}$ for all $\mathcal{C} \in \bigcup_{k = 1}^K C^{\bar{\delta}}_k(\mathcal{M})$. First, note that solving Problem~\ref{prob:rides}, under Approximation~\ref{approx},  with a simplified demand matrix  $D^\mathrm{rp} = D^{\mathcal{C},s}$ is essentially a shortest path problem between the nodes in the sequence $s\in\mathcal{S}_\mathcal{C}$. Instead of solving the shortest path problem for every bag of requests and for every valid sequence, it can be computed beforehand between every pair of nodes in the network. Since each instance has a worst-case computational complexity of $\mathcal{O} (\abs{\cV}^2)$, the overall computational complexity of all shortest path problems is $\mathcal{O} (\abs{\cV}^4)$. Second, consider a bag of requests $\mathcal{C} \in C_k^{\bar{\delta}}(\mathcal{M})$. For each sequence $s\in\mathcal{S}_\mathcal{C}$, leveraging the precomputation of the shortest path problem between every pair of nodes in the network, one has to compute $\delta^{\mathcal{C},s}_m, \; {\forall m\in \mathcal{C}}$ and $\tilde{J}(X^{\mathcal{C},s})$. The computation of $\delta^{\mathcal{C},s}_m, \; \forall m\in \mathcal{C}$ is the bottleneck, which involves $k$ matrix-vector multiplications of dimension $|\mathcal{V}|$, each with a computational complexity of $\mathcal{O}(|\mathcal{V}|^2)$. Thus, for each sequence $s\in\mathcal{S}_\mathcal{C}$, one has to perform floating-point operations with a computational complexity of $\mathcal{O}(k|\mathcal{V}|^2)$. Third, we proceed to upper bound the number of different valid sequences $s \in \mathcal{S}_\mathcal{C}$ with $\mathcal{C} \in C_k^{\bar{\delta}}(\mathcal{M})$. Consider a bag of requests $\mathcal{C} \in C_k^{\bar{\delta}}(\mathcal{M})$, then valid sequences $s\in\mathcal{S}_\mathcal{C}$ are permutations of the origin and destination nodes of $k$ requests such that i)~for each request the corresponding origin is visited before the destination; and ii)~no edge carries a null flow. Henceforth, consider that the requests in $\mathcal{C}$ are unique, otherwise the number of unique sequences is smaller.  Disregarding the two constraints above, there exist $(2k)!$ permutations of the $k$ origin and $k$ destination nodes of the requests in $\mathcal{C}$. Thus, $(2k)!$ can be used as an upper bound of $|\mathcal{S}_\mathcal{C}|$. It follows that, for each $\mathcal{C} \in C_k^{\bar{\delta}}(\mathcal{M})$, floating point operations with a complexity of $\mathcal{O}((2k!)k|\mathcal{V}|^2)$ must be performed. Fourth, the number of feasible request bags with $k$ elements, i.e., $|C_k^{\bar{\delta}}(\mathcal{M})|$, is upper bounded by the number of request bags with $k$ elements, i.e., $|C_k(\mathcal{M})|$, which amounts to the combinations of size $k$ of $M$ elements with repetition. In turn, it can be upper bounded by the number of permutations of size $k$ of $M$ elements with repetition, which is $M^k$. Therefore, the number of sequences that need to be checked to compute every feasible $D^{\mathcal{C},\star}$  is upper bounded by $(2k)!M^k$.  It follows that the computation of $D^{\mathcal{C},\star}$ for all $\mathcal{C} \in C_k^{\bar{\delta}}(\mathcal{M})$ has a computational complexity of $\mathcal{O}(M^k(2k!)k|\mathcal{V}|^2)$.  On top of that, given that the number of requests, $M$, grows with $\mathcal{O}(|\mathcal{V}^2|)$, then one can write that the computation of $D^{\mathcal{C},\star}$ for all $\mathcal{C} \in C_k^{\bar{\delta}}(\mathcal{M})$ has a computational complexity of $\mathcal{O}((2k!)k|\mathcal{V}|^{2k+2})$. The overall computational complexity of the computation of $D^{\mathcal{C},\star}$ for all $\mathcal{C} \in \bigcup_{k = 1}^K C^{\bar{\delta}}_k(\mathcal{M})$ grows with $\mathcal{O}((2K!)K |\mathcal{V}|^{2K+2})$ with $K$ and $|\mathcal{V}|$, since the complexity of the bags $\mathcal{C} \in C_K^{\bar{\delta}}(\mathcal{M})$ dominate over the others. To conclude, the complexity of the precomputation of the shortest path problem between every pair of nodes in the network is dominated and the overall computational complexity grows with $\mathcal{O}((2K!)K |\mathcal{V}|^{2K+2})$ with $K$ and $|\mathcal{V}|$.

\section{Proofs}

%\subsection{Proof of Lemma~A\ref{lem:sc}}\label{app:bound_Sc}
%%\subsection{Proof of asymptotic growth of $|\mathcal{S}_\mathcal{C}|$ with $\mathcal{C} \in C_k(\mathcal{M})$}\label{app:bound_Sc}
%
%Consider a bag of requests $\mathcal{C} \in C_k(\mathcal{M})$, then valid sequences $s\in\mathcal{S}_\mathcal{C}$ are permutations of the origin and destination nodes of $k$ requests such that i)~for each request the corresponding origin is visited before the destination; and ii)~no edge carries a null flow. Henceforth, consider that the request in $\mathcal{C}$ are unique, otherwise the number of unique sequences is smaller. First, disregarding the two constraints above, there exist $(2k)!$ permutations of the $k$ origin and $k$ destination nodes of the requests in $\mathcal{C}$. Second, we establish a lower bound on the number of sequences that do not satisfy condition i) above. There are $k(2k-1)!$ sequences that start with a destination, which do not follow i). There are $k(k-1)(2k-2)!$ sequences that start with an origin which is followed by an invalid destination, which do not follow i) either. The same argument can be used $k$ times to establish a lower bound of $k!\sum_{i = 1}^k (2k-i)!/(k-i)!$ sequences that do not satisfy i). Finally, note that $(2k)!-k!\sum_{i = 1}^k (2k-i)!/(k-i)!$ is an upper bound on $|\mathcal{S}_\mathcal{C}|$, which grows with $\mathcal{O}(k^kk!)$.

\subsection{Proof of Lemma~\ref{lemma:finalprob}}\label{app:1}

The time until the occurrence of a request follows an exponential distribution. Denote the distribution of the time until an occurrence of a request $i$ as $E_{i} \sim \mathrm{Exp}\left(\alpha_{i}\right) \forall i \in\{1, \ldots, k\}$. Moreover, consider the maximum and minimum distributions of these, i.e., $X=\max \left\{E_{1}, \ldots, E_{k}\right\}$ and $Y=\min \left\{E_{1}, \ldots, E_{k}\right\}$.
%$$
%\begin{aligned}
%	& X=\max \left(E_{1}, \ldots, E_{K}\right) \\
%	& Y=\min \left(E_{1}, \ldots, E_{K}\right) .
%\end{aligned}
%$$
The goal is to find the probability of all requests occurring within a maximum time interval $\bar{t}$, or equivalently the probability of the difference between the maximum and minimum bring less than or equal to $\bar{t}$. To that purpose, define $Z=X-Y$, whose cumulative density function $F_{Z}(t)$ allows to obtain the desired result as
%\begin{equation}\label{eq:p_t_bar}
%	P_{\bar{t}}\left(\alpha_{1}, \ldots, \alpha_{K}\right)=F_{Z}(\bar{t}).
%\end{equation}
%Rewriting \eqref{eq:p_t_bar} yields
\begin{equation}\label{eq:p_t_bar_extended}
	\begin{aligned}
		P_{\bar{t}}\left(\alpha_{1}, \ldots, \alpha_{k}\right) & =F_{Z}(\bar{t})=P(Z \leqslant \bar{t}) \\
		& =\int_{0}^{+\infty} \int_{0}^{y+\bar{t}} f_{XY}(x, y) \mathrm{d}x \mathrm{d}y,\\
	\end{aligned}
\end{equation}
where $f_{XY}(x, y)$ denotes the joint probability density function of $X$ and $Y$. In what follows, we are going to compute  an explicit expression for $f_{XY}(x, y)$. The joint cumulative density function of $x, y$ is denoted by $F_{XY}(x, y)$, i.e., $F_{XY}(x, y)=P(Y \leqslant y, X \leqslant x)$.
%\begin{equation}\label{eq:cumm_df}
%	F_{XY}(x, y)=P(Y \leqslant y, X \leqslant x).
%\end{equation}
Note that $Y\leq y$ and $X \leq x$ if all the realizations are less than or equal to $x$ and at least one is less than or equal to $y$. Thus, %As a result, \eqref{eq:cumm_df} is rewritten as
\begin{equation*}%\label{eq:expand_P}
P(Y\! \leqslant \! y, X \!\leqslant \! x)= P\Bigg(\underbrace{\!\!\!\left(\bigcup_{i=1}^{k}\!\left\{E_{i} \leqslant y\right\}\right)}_{\substack{\text{at least one less}\\\text{than or equal to } y}} \! \cap \! \underbrace{\left(\bigcap_{i=1}^{k}\!\left\{E_{i} \leqslant x\right\}\right)}_{\substack{\text {all less than } \\ \text{or equal to }x}}  \!\!\!\Bigg)\!.
\end{equation*}
%\begin{equation}\label{eq:expand_P}
%	\begin{split}
%		&P(Y \leqslant y, X \leqslant x)= \\
%		&\quad \quad \quad P\Bigg(\underbrace{\left(\bigcup_{i=1}^{K}\left\{E_{i} \leqslant y\right\}\right)}_{\substack{\text{at least one}\\\text{less than } y}} \cap \underbrace{\left(\bigcap_{i=1}^{K}\left\{E_{i} \leqslant x\right\}\right)}_{\text {all less than } x}  \Bigg).
%	\end{split}
%\end{equation}
%\begin{equation}\label{eq:expand_P}
%F_{XY}(x, y) =  
%		 P\Bigg(\underbrace{\!\!\!\left(\bigcup_{i=1}^{K}\!\left\{E_{i} \leqslant y\right\}\right)}_{\substack{\text{at least one}\\\text{less than } y}} \! \cap \! \underbrace{\left(\bigcap_{i=1}^{K}\!\left\{E_{i} \leqslant x\right\}\right)}_{\text {all less than } x}  \!\!\!\Bigg).
%\end{equation}
For any two sets $A, B$ it follows that $A \cap B=A \backslash \bar{B}$ and $P(A \backslash \bar{B})=P(A)-P(A \cap \bar{B})$, where $\bar{B}$ denotes the complement of set $B$. Making the correspondence $P(Y \leq y, X \leqslant x) =P(A \cap B)$ with $A=\bigcap_{i=1}^{k}\left\{E_{i} \leqslant x\right\}$ and $B=\bigcup_{i=1}^{k}\left\{E_{i} \leqslant y\right\}$, it follows that % \eqref{eq:expand_P} yields
%$$
%\begin{aligned}
%	P(Y \leqslant y&, X \leqslant x) = \\
%	& = P\left(\bigcap_{i=1}^{K}\left\{E_{i} \leqslant x\right\}\right)-P\left(\bigcap_{i=1}^{K}\left\{y<E_{i} \leqslant x\right\}\right) \\
%	& =\prod_{i=1}^{K} P\left(E_{i} \leq x\right)-\prod_{i=1}^{K} P\left(y<E_{i} \leqslant x\right) \\
%	& =\prod_{i=1}^{K} F_{E_{i}}(x)-\prod_{i=1}^{K}\left(F_{E_{i}}(x)-F_{E_{i}}(y)\right),
%\end{aligned}
%$$
\begin{equation}\label{eq:explicit_FXY}
	\begin{aligned}
		F_{XY}(x, y) & =\prod_{i=1}^{k} P\left(E_{i} \leq x\right)-\prod_{i=1}^{k} P\left(y<E_{i} \leqslant x\right) \\
		& =\prod_{i=1}^{k} F_{E_{i}}(x)-\prod_{i=1}^{k}\left(F_{E_{i}}(x)-F_{E_{i}}(y)\right), 		%P(Y \leqslant y, X \leqslant x)
	\end{aligned}
\end{equation}
where the independence of $E_{1}, \ldots, E_{k}$ was employed. Since, $f_{XY}(x, y)=\frac{\partial^{2}}{\partial x \partial y} F_{XY}(x, y)$ and $f_{XY}(x, y)=\frac{\partial^{2}}{\partial x \partial y} F_{XY}(x, y)=0$ for $x <x y$, one may write \eqref{eq:p_t_bar_extended} as
%\begin{equation}\label{eq:Fz}
%	\begin{aligned}
	%		F_{Z}(\bar{t}) & =\int_{0}^{+\infty} \int_{y}^{y+\bar{t}} f_{XY}(x, y) \mathrm{d} x \mathrm{d} y \\
	%		& =\int_{0}^{+\infty}\left[\frac{\partial F_{XY}(x, y)}{\partial y}\right]_{x=y}^{x=y+\bar{t}} \mathrm{d} y.
	%	\end{aligned}
%\end{equation}
\begin{equation}\label{eq:Fz}
		F_{Z}(\bar{t}) =\int_{0}^{+\infty}\left[\frac{\partial F_{XY}(x, y)}{\partial y}\right]_{x=y}^{x=y+\bar{t}} \mathrm{d} y.
\end{equation}
%where the fact that $f_{XY}(x, y)=\frac{\partial^{2}}{\partial x \partial y} F_{XY}(x, y)=0$ for $x \leqslant y$ was used. 
Since, $E_1,\ldots,E_k$ are exponential distributions, the cumulative functions are given by $F_{E_i}(x) = 1-e^{-\alpha_i x}, \; x\geq 0$. Thus, for $x>y$, from \eqref{eq:explicit_FXY} we obtain
\begin{equation*}
	\frac{\partial F_{XY}(x, y)}{\partial y}=\sum_{i=1}^{k}\left(\alpha_{i} e^{-\alpha_{i} y} \prod_{\substack{j=1 \\ j \neq i}}^{k}\left(e^{-\alpha_{j} y}-e^{-\alpha_{j} x}\right)\right)
\end{equation*}
and, after some algebraic manipulation, 
%\begin{equation}\label{eq:dF_expanded}
%	\begin{aligned}
%		& \left[\frac{\partial F_{x y}(x, y)}{\partial y}\right]_{x=y}^{x=y+\bar{t}} = \\
%		& \quad \quad    = \sum_{i=1}^{K}\left(\alpha_{i} e^{-\alpha_{i} y} \prod_{\substack{j=1 \\
%				j \neq i}}^{K}\left(e^{-\alpha_{j} y}-e^{-\alpha_{j} y-\alpha_{j} \bar{t}}\right)\right) \\
%		&   \quad \quad   \quad \quad-\sum_{i=1}^{K}\left(\alpha_{i} e^{-\alpha_{i y}} \prod_{\substack{j = 1 \\
%				j \neq i}}^{K}\left(e^{-\alpha_{j} y}-e^{-\alpha_{j} y}\right)\right) \\
%		&  \quad \quad    =\sum_{i=1}^{K}\left(\alpha_{i} e^{-\alpha_{i} y} \prod_{\substack{j=1 \\
%				j \neq i}}^{K}\left(e^{-\alpha_{j} y}\right) \prod_{\substack{j=1 \\
%				j \neq i}}^{K}\left(1-e^{-\alpha_{j} \bar{t}}\right)\right) \\
%		%& =\sum_{i=1}^{K}\left(\alpha_{i} e^{-\alpha_{i} y} e^{\sum_{\substack{j=1 \\ j\neq 1}}^{K}\left(\alpha_{j}\right) y} \prod_{j=1}^{K}\left(1-e^{-\alpha_{j} \bar{t}}\right)\right) \\
%		& \quad \quad     =\sum_{i=1}^{K}\left(\alpha_{i} e^{-\sum_{j=1}^{K} \alpha_{j} y} \prod_{\substack{j=1 \\
%				j \neq i}}^{K}\left(1-e^{-\alpha_{j} \bar{t}}\right)\right)
%	\end{aligned}
%\end{equation}
\begin{equation}\label{eq:dF_expanded}
\left[\frac{\partial F_{XY}(x, y)}{\partial y}\right]_{x=y}^{x=y+\bar{t}}  \!\!\!\!= \!  \sum_{i=1}^{k}\!\left( \! \alpha_{i} e^{-\sum_{j=1}^{k} \alpha_{j} y} \prod_{\substack{j=1 \\	j \neq i}}^{k}\!\!\left(1\!-\!e^{-\alpha_{j} \bar{t}}\right)\right).
\end{equation}
Thus, from \eqref{eq:Fz} and \eqref{eq:dF_expanded} it follows that
\begin{equation*}
	F_{Z}(\bar{t})  =\sum_{i=1}^{k}\left(\frac{\alpha_{i}}{\sum_{j=1}^{k} \alpha_{j}} \prod_{j=1}^{k}\left(1-e^{-\alpha_{j} \bar{t}}\right)\right).
\end{equation*}

\qed

\subsection{Proof of Theorem~\ref{theorem:one}}\label{app:2}

%${\!\abs{\cM}(\abs{\cM}\!-\!1)\!}$ 

First, we address the bound on the iterations required for the convergence of Algorithm~\ref{alg:one}. If for each bag $\mathcal{C}$ chosen in each iteration we set  $\Delta \tilde{J}_{\mathcal{C}}  = 0$, then bag $\mathcal{C}$ will not be chosen again. Since $\abs{C_k^{\bar{\delta}}} \leq |\mathcal{M}|^k$, then  $\abs{\bigcup_{k = 1}^K C^{\bar{\delta}}_k(\mathcal{M})}\leq \abs{\cM}(\abs{\cM}^K-1)/(\abs{\cM}-1)$. The optimality of the vehicle flows $\gamma_\mathcal{C}$ is carried out making use of an analogy with the continuous Knapsack problem,  which can be solved by a polynomial-time greedy algorithm~\cite{Dantzig1957}. The algorithm consists allocating, at every iteration, the maximum amount of the resource with the highest improvement in the objective function per unit of  the resource. Similarly to the continuous Knapsack problem, the goal is to minimize $\tilde{J}(X^\star_\gamma)$ by allocating user flow $|\mathcal{C}|{\gamma_{\mathcal{C}} \geq 0}$ with $\mathcal{C} \in \bigcup_{k = 1}^K C^{\bar{\delta}}_k(\mathcal{M})$. In this setting, the vehicle flow ${\gamma_{\mathcal{C}} \geq 0}$ corresponds to a  ride-pooled user flow of $|\mathcal{C}|{\gamma_{\mathcal{C}} \geq 0}$, which is upper bounded as $\sum_{\mathcal{C} \in \bigcup_{k = 1}^K C^{\bar{\delta}}_k(\mathcal{M})} |\mathcal{C}|\gamma_{\mathcal{C}} \leq \sum_{i\in \mathcal{M}} \alpha_{i}$. First, borrowing the notation from Section~\ref{sec:SD}, if $\gamma_{\mathcal{C}}$ is assigned to bag $\mathcal{C}$, then the corresponding decrease in the cost function amounts to $\tilde{J}(\gamma_{\mathcal{C}}X^{\mathcal{C},0})\! -\!\tilde{J}(\gamma_{\mathcal{C}}X^{\mathcal{C},\star}) =\gamma_{\mathcal{C}}(\tilde{J}(X^{\mathcal{C},0})\! -\!\tilde{J}(X^{\mathcal{C},\star})) = \gamma_{\mathcal{C}}\Delta{\tilde{J}}_{\mathcal{C}}$, where the linearity of $\tilde{J}(X)$ w.r.t.\ $X$ played a key role. Thus, the allocation of $|\mathcal{C}|\gamma_{\mathcal{C}}$ leads to an improvement on the cost that amounts to $\Delta{\tilde{J}}_{\mathcal{C}}/|\mathcal{C}|$ per unit of user flow. Note that this relative improvement does not depend on the allocation of other bags, thus, since the user flow is upper bounded, the optimal solution is to allocate the maximum user flow to the bags that allow for the highest relative improvement w.r.t.\ the user flow. Second, the maximum user flow for $\mathcal{C}$ corresponds to the maximum $\gamma_{\mathcal{C}}$ that can be allocated. As pointed out in Section~\ref{sec:alg}, throughout the algorithm, $\alpha^\prime_m$ corresponds to the demand of $r_m$ which has not yet been ride-pooled with another request. Thus, the value of $\gamma_{\mathcal{C}}$ that can be allocated has an upper bound given by $\gamma_{\mathcal{C}} \leq \min\left(\alpha^\prime_m/m_{\mathcal{C}}(m), m\in \mathcal{C}\right) P_{\bar{t}}(\alpha_m^\prime, m\in \mathcal{C}) $. Note that Algorithm~\ref{alg:one} corresponds to allocating the maximum amount of $\gamma_{\mathcal{C}}$, where $\mathcal{C}$ is such that, at each iteration, the highest positive relative improvement in the objective function w.r.t.\ the user flow is achieved, i.e., $ \mathcal{C} \in \mathrm{argmax}_{\mathcal{C} \in \bigcup_{k = 1}^K  C^{\bar{\delta}}_k(\mathcal{M})} (\Delta \tilde{J}_{\mathcal{C}}/|\mathcal{C}|)$, which shows its optimality. \qed

\begin{IEEEbiography}[{\includegraphics[width=1in,height=1.25in,clip,keepaspectratio]{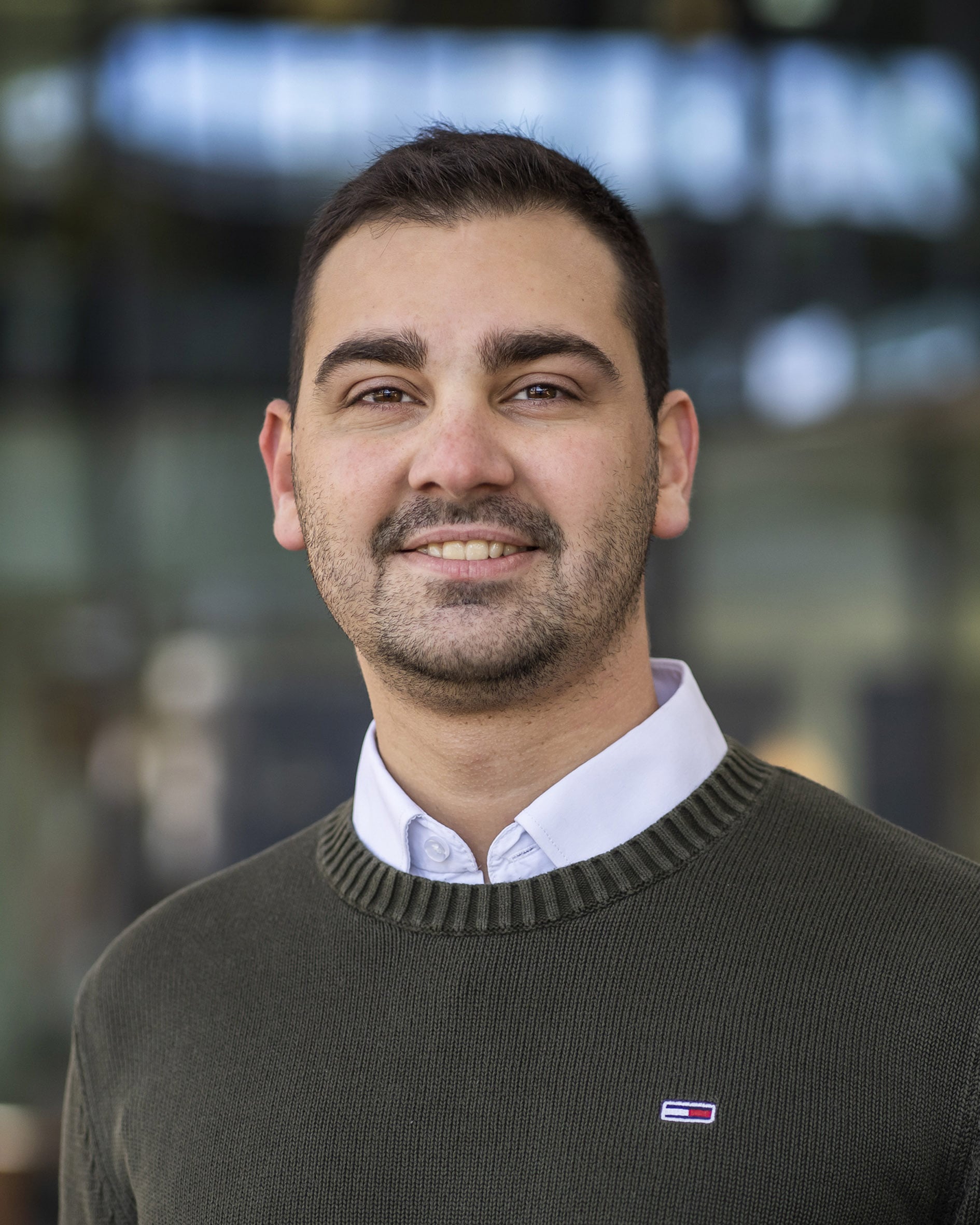}}]{Fabio Paparella} is a Ph.D. student in the Control Systems Technology (CST) section at Eindhoven University of Technology, The Netherlands. He studied mechanical engineering at Politecnico di Milano, Italy, where he received his Bachelor's degree in 2017 and his Master's cum laude in 2020 with a thesis in collaboration with NASA Jet Propulsion Laboratory, California, USA. His research interests include mobility-on-demand, smart mobility, and optimization.
\end{IEEEbiography}
\begin{IEEEbiography}[{\includegraphics[width=1in,height=1.25in,clip,keepaspectratio]{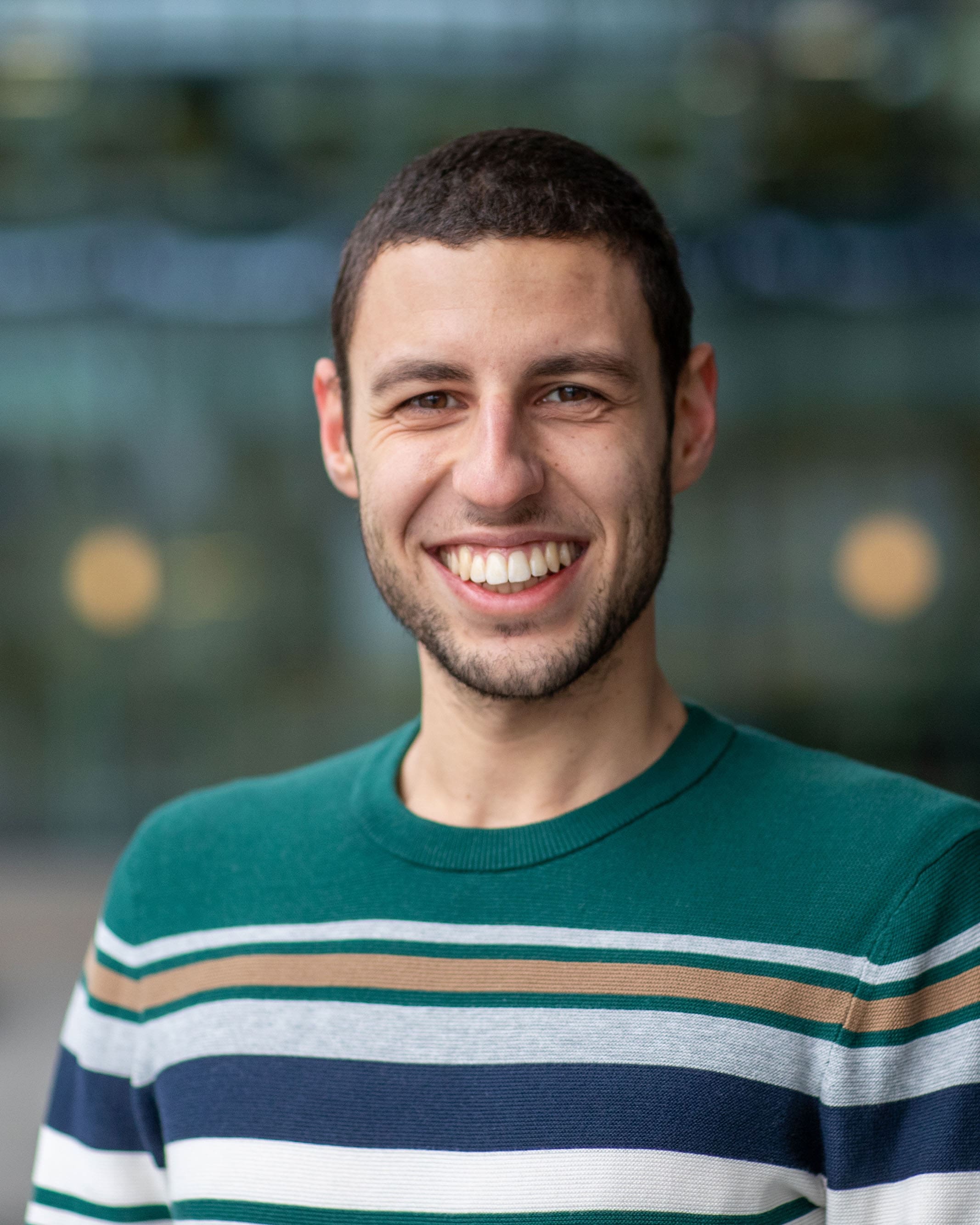}}]{Leonardo Pedroso} obtained the M.Sc. in Aerospace Engineering in 2022 from Instituto Superior Técnico, University of Lisbon (ULisboa), Portugal. He has received four outstanding academic performance awards, endowed by ULisboa. From 2019 to 2022, he was a member of a research team funded by the Portuguese National Science Foundation. Since 2023, he has been a Ph.D. Candidate in the Control Systems Technology section at Eindhoven University of Technology, The Netherlands. His research interests include fair and sustainable mobility, aggregative game theory, and decentralized control and estimation over very large-scale networks.
\end{IEEEbiography}
\begin{IEEEbiography}[{\includegraphics[width=1in,height=1.25in,clip,keepaspectratio]{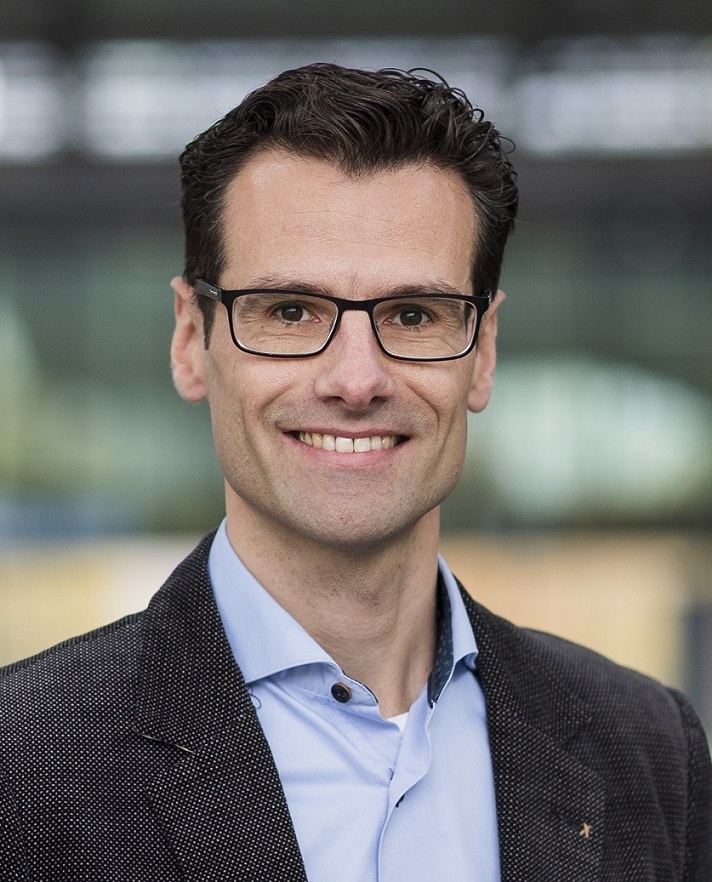}}]{Theo Hofman} was born in Utrecht, The Netherlands, in 1976. He received his M.Sc. (with honors) in 1999 and Ph.D. degree in 2007, both in Mechanical Engineering from Eindhoven University of Technology, Eindhoven. From 1999 to 2003, he was a researcher and project manager with the R\&D Department of Thales Cryogenics B.V., Eindhoven, The Netherlands. From 2003 to 2007, he was a scientific researcher at Drivetrain Innovations B.V., Eindhoven. Since 2010, he has been an Associate Professor with the Control Systems Technology group. His research interests are system design optimization methods for complex dynamical engineering systems and discrete topology design using computational design synthesis.
\end{IEEEbiography}
\begin{IEEEbiography}[{\includegraphics[width=1in,height=1.25in,clip,keepaspectratio]{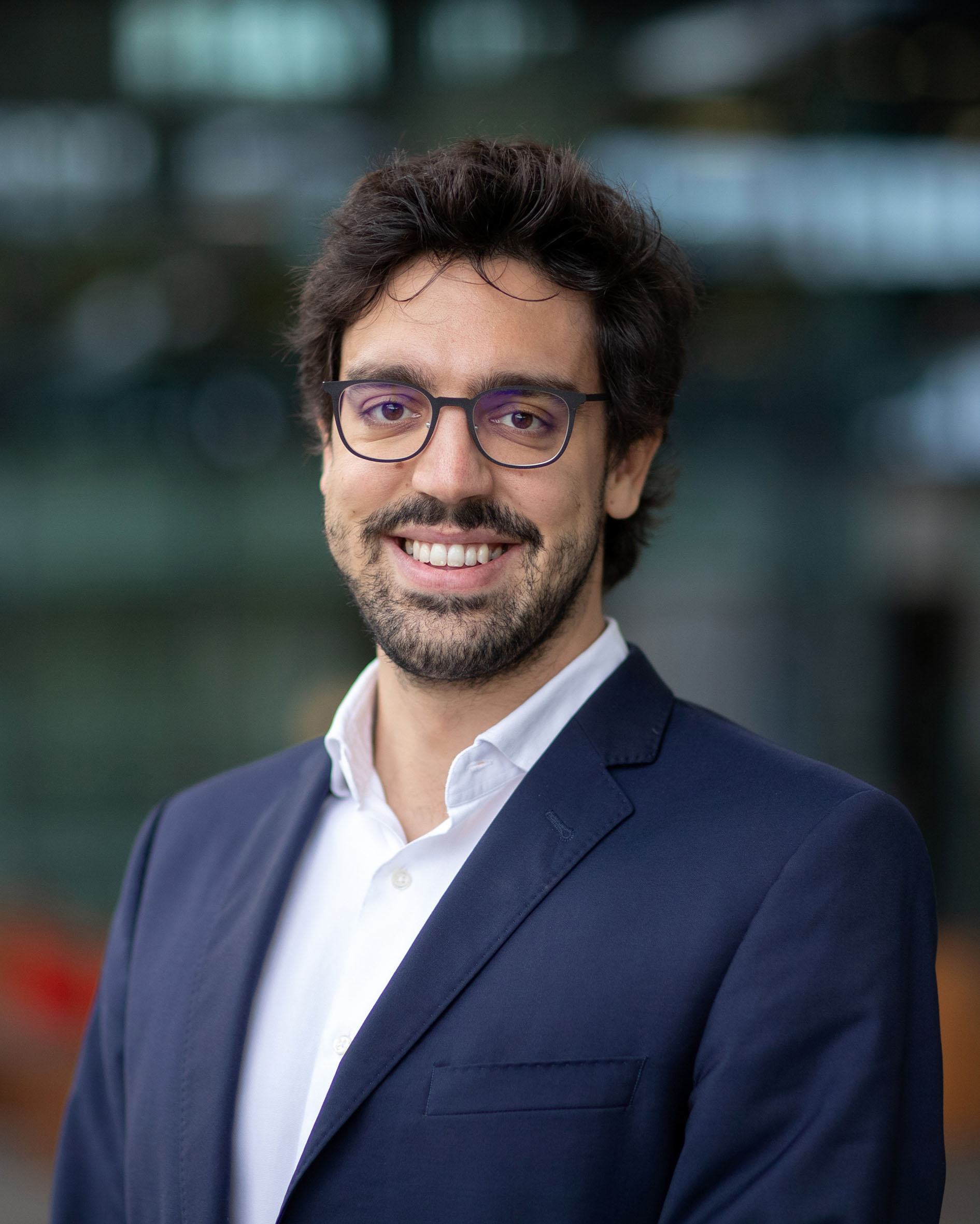}}]{Mauro Salazar} is an Assistant Professor in the Control Systems Technology section at Eindhoven University of Technology (TU/e), and co-affiliated with Eindhoven AI Systems Institute (EAISI). He received the Ph.D. degree in Mechanical Engineering from ETH Zurich in 2019. Before joining TU/e he was a Postdoctoral Scholar in the Autonomous Systems Lab at Stanford University.
Dr. Salazar’s research is focused on optimization models and methods for cyber-socio-technical systems design and control, with a strong focus on sustainable mobility.Both his Master thesis and PhD thesis were recognized with the ETH Medal, and his papers were granted the Best Student Paper award at the 2018 Intelligent Transportation Systems Conference and at the 2022 European Control Conference.
\end{IEEEbiography}
\vfill

\end{document}